\documentclass[preprint,number]{elsarticle}
\usepackage[latin1]{inputenc}
\usepackage[english]{babel}
\usepackage{fancyhdr}
\usepackage{amsfonts}
\usepackage{bbm}
\usepackage{pdflscape}
\usepackage{amssymb,amsmath,mathrsfs,centernot}								
\usepackage{multicol}									
\usepackage{graphicx}									
\usepackage{xcolor}										
\usepackage[bottom]{footmisc} 				
\usepackage{hyperref}									
\usepackage[nottoc,numbib]{tocbibind}
\usepackage{floatrow}
\usepackage{verbatim}
\usepackage{rotating}
\usepackage[labelfont={footnotesize,bf},textfont={footnotesize}]{caption}
 \usepackage[toc,page]{appendix}
  \usepackage{caption}
\usepackage{subcaption}
\usepackage{amsthm}
\usepackage{pgfplots}
\usepackage{tikz}
\usepackage{fancyhdr}

\usetikzlibrary{decorations.pathreplacing,arrows.meta}
\hypersetup{colorlinks=true, citecolor=blue, linkcolor=blue, urlcolor=blue}	
\makeatletter
\def\ps@pprintTitle{%
  \let\@oddhead\@empty
  \let\@evenhead\@empty
  \let\@oddfoot\@empty
  \let\@evenfoot\@oddfoot
}

\makeatother

\newcommand{\dd}{ \text{d}} 	

\newtheorem{theo}{Theorem}[section]
\newtheorem{definition}{Definition}[section]
\newtheorem{prop}{Proposition}[section]

\newtheorem{lem}{Lemma}[section]

\begin{document}

\begin{frontmatter}

\title{Dark Markets with Multiple Assets: Segmentation, Asymptotic Stability, and Equilibrium Prices }

\author[udes]{Alain B\'{e}langer }
\ead{alain.a.belanger@usherbrooke.ca}
\author[uotwam]{Ndoun\'{e} Ndoun\'{e}}
\ead{nndoune@uottawa.ca}
\author[uotwa]{Roland Pongou}
\ead{rpongou@uottawa.ca; rpongou@hsph.harvard.edu}

\address[udes]{\'{E}cole de Gestion, Universit\'{e} de Sherbrooke, Qu\'{e}bec, J1K2R1, Canada}
\address[uotwam]{Department of Mathematics and Statistics, University of Ottawa, K1N 6N5, Canada}
\address[uotwa]{Department of Economic Science, University of Ottawa, K1N 6N5, Canada; Harvard T.H. Chan School of Public Health,
Boston, MA 02115, US}



\begin{abstract}
We study a generalization of the model of a dark market due to Duffie-G\^{a}rleanu-Pedersen \cite{Duffie2005}. Our market is segmented and involves multiple assets. We show that this market has a unique asymptotically stable equilibrium. In order to establish this result, we use a novel approach inspired by a theory due to McKenzie and Hawkins-Simon. Moreover, we obtain a closed form solution for the price of each asset at which investors trade at equilibrium. We conduct a comparative statics analysis which shows, among other sensitivities, how equilibrium prices respond to the level of interactions between investors.
\end{abstract}


\end{frontmatter}


\section{Introduction }

 An over-the-counter (OTC) market is a decentralized market, without a physical exchange place, where investors trade privately. Trades may take place through telephone or online, and transactions are concluded bilaterally between any two parties. Some OTC markets allow market-makers and traders.
As mentioned in Duffie, G\^arleanu and Pederson \cite{Duffie2007}, several types of assets, such as mortgage-backed securities, swaps and many other types of derivatives, emerging-market debt, corporate bonds, government bonds, bank loans, private equity, and real estate, are traded in OTC markets. A trade between a seller and a buyer is executed when they agree on the price. In general, OTC markets are less transparent than stock markets, which is why they are also called $dark$ $markets$  (see \cite{Duffie2012}).


In this paper, we study an extension of the OTC market that is segmented and involves multiple assets.
 An OTC market with $K$ assets is said to be $partially$ $segmented$ (or simply $segmented$) if buyers have to decide, prior to entering the market, which one of the $K$ assets they want to buy. Each agent  holds at most one unit of any asset $i$ and cannot short-sell. The partially segmented OTC market models were introduced by Vayanos and Wang \cite{Vayanos2007} for the case of two assets with different liquidity assumptions. Another model of a segmented OTC market is studied by Weill \cite{Weill2008}. These two models appear as the first extensions of the Duffie-G\^arleanu-Pederson OTC model (see \cite{Duffie2005} and \cite{Duffie2007}).
In a partially segmented market, each buyer looks for a specific asset, then searches for a potential seller of this asset. \newline For example, if an investor wants a particular corporate bond, he enters the market and looks for a seller of the corporate bond no matter if, at this moment, another similar corporate bond may be more liquid than the first one. This reflects the opacity of information; the prices are unpublished and the other agents in the market are unaware of the price at which the asset is traded.
 A market can be segmented in order to reduce unfair competition between investors. The segmentation of a market allows each investor to get a private monopoly of a given asset in order to increase the price from the socially optimal price and therefore sell at a higher price.

 Segmented OTC markets with multiple assets are highly prevalent in real life, but they have not been sufficiently studied.
Yet understanding the functioning of these markets has practical implications, as investors seek investment strategies that employ sophisticated asset allocations and in-depth research which aims at supplying specific results
adapted to their needs. This is a great motivation among so many others to study
this class of markets. In this paper, our goal is threefold. $First$, we study the asymptotic stability of these markets, which is important to understand their long-run predictions. $Second$, we derive a closed form solution for the price of each
asset at which investors trade at equilibrium. $Third$, we conduct a comparative statics analysis of the price of each asset, and characterize the cross elasticity of the demand.

In our model, each investor's liquidity is characterized  by two states, namely a ``high" and a ``low'' state. Also each investor either owns one unit of an asset $i$ or does not own it.
A transaction occurs when a high type investor who does not own a specific asset meets a seller of this asset, that is, a low liquidity type investor owning the asset (see \cite{Beland2016}).
  We assume that the assets are differentiated by their liquidity and intrinsic value, so that we can compare the preferences of investors in relation to each asset and liquidity. As in Duffie, G\^arleanu and Pederson \cite{Duffie2005}, we assume that each agent's utility function is linear and depends on his type, his wealth and time. An asset may be hit by a valuation shock, with these shocks being independent across assets.

 Our first result is to show that any partially segmented OTC market with multiple assets has a unique asymptotically stable equilibrium (Theorem 3.1). This finding is essential for the understanding of the long-run behavior of this class of markets. Indeed, it implies that any market state is transitory, unless it coincides with the equilibrium state. This also means that any market state that starts near the equilibrium moves closer to it as time elapses, and any displaced motion gets back to the equilibrium. We should also note that, in addition to the substantive finding, we make a methodological contribution in the sense that we rely on new arguments to establish our result. In fact, due to the existence of multiple assets in the market, the method used in the pioneering work of Duffie, G\^arleanu and Pederson \cite{Duffie2005} cannot be extended to our environment. Our argument uses a theory developed by Mckenzie \cite{Mckenzie2009} and Hawkins and Simon \cite{Garleanu1949} and exploit the properties of Hurwitz determinants. Our methodology also allows us to derive the equilibrium timing of the seller, which is the expected number of days that an asset stays in the market before it is sold \footnote{\textit{((In the case of one asset, \cite{Duffie2007} show that the equilibrium price can be obtained as the subgame perfect Nash equilibrium of a Rubinstein-type alternating-offers game.))}}.

 Our second result gives a closed-form solution of the price at which each asset is traded at equilibrium (Theorem 4.1). We note that the price of each asset depends on its exogenous characteristics and on the characteristics of other assets. It also depends on the level of interactions between agents. More formally, the price of a given asset can be expressed as a quotient of polynomial functions of the exogenous parameters of the model and the unique steady state equilibrium. We provide an example that shows how prices are determined. Importantly, it should be noted that because these prices are equilibrium prices, small shocks to their exogenous determinants can temporarily move them, but they will ultimately return to their equilibrium level.

 We are also interested in studying the rate of change of the equilibrium prices when some exogenous parameters of the model grow simultaneously to infinity. We find that the limit prices, taken as the asymptotic limit of exogenous parameters, may not exist (Proposition 6.1). Such a situation could not happen in Duffie, G\^arleanu and Pederson \cite{Duffie2007}, because the parameters of the segmentation were not considered in the seminal model.
 One important result shows that there are two kinds of exogenous parameters, the first being the type of parameters whose variation maintains the existence of equilibrium prices and the second being the type of parameters that express the segmentation which break it off. By changing randomly these parameters, the equilibrium prices can be broken when the parameters grow to infinity. It is noteworthy that the corresponding result shows that the variation of the intensities of meetings maintains the equilibrium prices (Proposition 6.1).

  Our last result characterizes the cross elasticity of the demand, for equilibrium prices, that is, the responsiveness of the variation of an exogenous parameter for a given asset to change in the price of other assets (Proposition 6.2). Indeed, the cross elasticity analysis shows that, depending on the values of the other parameters, the price of an asset increases or decreases in the frequency at which agents meet for the exchange of that asset. Surprisingly, there exists, for each given asset, a value of its frequency of meeting which makes its price constant. We provide a simple example that illustrates our theoretical findings.
Because these findings are specific to markets involving multiple assets, they constitute a useful addition to those found in the classical papers of Duffie, G\^arleanu and Pederson (\cite{Duffie2005},\cite{Duffie2007}) and Duffie \cite{Duffie2012}.

The rest of the paper is organized as follows. In section 2, we present the model of a segmented OTC market with multiple assets. In section 3, we show that any segmented OTC market with multiple assets has a unique steady state which is asymptotically stable. In section 4, we obtain a closed form solution for the equilibrium price of each asset. We present a numerical illustration of our analysis in section 5.
In section 6, we conduct a comparative statics analysis. Finally, section 7 concludes with a discussion of directions for further research. For clarity of exposition, we collect all the proofs in the appendix.

\section{A Model of a Segmented Dark Market with Multiple Assets}

In their paper \cite{Duffie2005},  Duffie, G\^{a}rleanu and Pedersen present their model of OTC market with one traded asset as a dynamical system of four differential equations with two constraints which can be reduced to a system of two differential equations with two constraints. In this section, we briefly present our model, inspired by their model, with $K$ traded assets, $K\geq 1,$ and segmentation.
Following Duffie \cite{Duffie2012}, let us consider a probability space $(\Omega,\mathcal{F},\mathrm{P})$ and $\{\mathcal{F}_{t},\: t\geq 0\}$ a filtration, that is the time $t$ information set, which is the $\sigma$-algebra of events that are known to the market participants at time $t$. The filtration models the evolution of information as it becomes available over time. Our OTC financial market model is based on the observation of a continuous time stochastic vector process expressing asset price processes. The components of the vector process could include for instance, security prices, interest rates, indicators for certain political events, insurance claims, and trade balance, to name just a fiew. For more details we refer the readers to \cite{Platen2006}.
In our market, there are two kinds of investors: buyers and sellers, who consume a single nonstorable
good that is used as a numeraire. We do not consider OTC market models with market-makers in this study.

The set of available assets will be denoted $\mathcal{I} = \{1,...,K\}$. Investors can hold at most one unit of any asset $i\in\mathcal{I}$ and cannot short-sell. Time is treated continuously and runs forever. The market is populated by a continuum of investors. At each time, an investor is characterized by whether he owns the $i$-th asset or not, and by an intrinsic type which is either a 'high' or a 'low' liquidity state. Our interpretation of liquidity state is the same as in \cite{Duffie2005}. For example, a low-type investor who owns an asset may have a need for cash and thus wants to liquidate his position. A high-type investor who does not own an asset may want to buy the asset if he has enough cash. Through time, investors' ownerships will switch randomly because of meetings leading to trades and the investor's intrinsic type will change independently  via an autonomous movement which can be considered as an idiosyncratic shock. This dynamics of an investor's type change is modeled by a (non-homogeneous) continuous-time Markov chain $Z(t)$ on the finite set of states $E$. The  state of  an investor is given by an element of $E = \{(l,n),(hi,o),(hi,n),(li,o), \, i\in\mathcal{I}\}$, where the first letter designates the investor's intrinsic liquidity state and the second designates whether the investor owns the asset $i$ or not.
If an investor initially does not own any asset and is a low-type, the switching intensity of becoming a high-type is denoted $\widetilde{\gamma}_{ui}$ as it depends on the asset type. If he initially does not own any asset but is a high-type, his switching intensity of becoming a low-type is denoted $\widetilde{\gamma}_{di}$. However, if an investor initially owns the specific asset $i$ and is a high-type, the switching intensity of becoming a low-type is denoted by $\gamma_{di}$. Finally, if he initially owns a specific asset $i$ but is a low-type, the switching intensity of becoming a high-type is $\gamma_{ui}$. We make the liquidity switches depend on the asset because these assets could have different  purchase prices and  dividend flows.

Let $\delta_{hi}> 0$ be the dividend flow from asset $i.$
 A low-type investor, who owns the asset $i$, has a holding cost $\delta_{di},$ with $\delta_{hi}> \delta_{di}> 0$. If $\theta_{i}(t)$ denote the ownership process for asset $i$, i.e. $\theta_{i}(t)=1$ if $Z(t)\in \{
(hi,o),(li,o)\}$ and $\theta_{i}(t)=0$ if $Z(t)\in \{
(hi,n),(l,n)\}$.


We define the asset value process as follows:
\begin{align}
  d A(t) & \triangleq  \sum_{i\in\mathcal{I}}\left[\theta_i(t)\left(\delta_{hi}-\delta_{di}\mathbbm{1}_{\{Z(t)=(li,o)\}}\right)d t - P_i(t) d \theta_i(t) \right]
\end{align}
where $P_i(t)$ is the price of the $i$th asset. This process is part of the wealth process which will be used in the price derivation.

In addition, investors meet each other randomly at a Poisson rate $\lambda_i$, but an exchange of the asset occurs only if an investor of type $(li,o)$ meets an investor of type $(hi,n)$.
One should notice that, without changes of positions, the system would stop after a finite time and the market would become inefficient.
 At any given time $t$, let $\mu_t(z)$ denote the proportion of investors in state $z \in E$. So, for each $t \geq 0$, $\mu_t$ is a probability law on $E$.

Let $m_i$ denote the proportion of asset $i$, with $i \in \mathcal{I}$. The dynamical system of investors' type proportion measure $\mu_t(z)$ for each $z\in E$,  consists of the system of equations:
\begin{eqnarray}
	\dot{\mu}_t(hi,n) &=& -\lambda_i \mu_t(hi,n)\mu_t(li,o) + \widetilde{\gamma}_{ui} \mu_t(l,n) - \widetilde{\gamma}_{di}\mu_t(hi,n), \ \forall i\in\mathcal{I}\label{peq1}\\
	\dot{\mu}_t(l,n) &=&  \sum_{i\in\mathcal{I}} \lambda_i \mu_t(hi,n) \mu_t(li,o)- \sum_{i\in\mathcal{I}}\widetilde{\gamma}_{ui} \mu_t(l,n) + \sum_{i\in\mathcal{I}}\widetilde{\gamma}_{di}\mu_t(hi,n)\label{peq2}\\
	\dot{\mu}_t(hi,o) &=& \lambda_i \mu_t(hi,n)\mu_t(li,o) + \gamma_{ui}\mu_t(li,o) - \gamma_{di}\mu_t(hi,o), \ \forall i\in\mathcal{I}\label{peq3}\\
	\dot{\mu}_t(li,o) &=& -\lambda_i \mu_t(hi,n)\mu_t(li,o) - \gamma_{ui} \mu_t(li,o) + \gamma_{di}\mu_t(hi,o), \ \forall i\in\mathcal{I}\label{peq4}
\end{eqnarray}
with the $K+1$ constraints

\begin{align*}
\mu_t(hi,o) + \mu_t(li,o) &= m_i, \ \forall i\in\mathcal{I}\\
\sum_{i\in\mathcal{I}}m_i  + \sum_{i\in\mathcal{I}} \mu_t(hi,n)  + \mu_t(l,n) &= 1
. \end{align*}

Following \cite{Beland2016}, the above dynamical system can be reduced to the following equivalent system of $2K$ equations:
\begin{align}\label{eq:pmasterSystem}
\begin{split}
	\dot{\mu}_t(hi,n) &= -\lambda_i \mu_t(hi,n)\mu_t(li,o) + \widetilde{\gamma}_{ui} \mu_t(l,n) - \widetilde{\gamma}_{di}\mu_t(hi,n), \ \forall i\in\mathcal{I}\\
	\dot{\mu}_t(li,o) &=  -\lambda_i \mu_t(hi,n)\mu_t(li,o) - \gamma_{ui}\mu_t(li,o) + \gamma_{di}\mu_t(hi,o), \ \forall i\in\mathcal{I}
\end{split}
\end{align}
with the $K+1$ constraints
\begin{align}
	\mu_t(hi,o) + \mu_t(li,o) &= m_i, \ \forall i\in\mathcal{I}\label{pconstraint1}\\
	\sum_{i\in\mathcal{I}} m_i + \sum_{i\in\mathcal{I}} \mu_t(hi,n)  + \mu_t(l,n) &= 1\label{pconstraint2}
\end{align}


\section{The Stability of Dark Markets}
\subsection{ Definitions and Concepts}

 We recall that a steady state of a dynamical system is a state at which the behavior of the system is unchanging in time. A steady state is $stable$ when the system always returns to the steady state after small disturbances. If for all initial values, the nearby integral curves all converge towards a steady state solution as $t$ increases, then the steady state is said to be $asymptotically$ $stable$. If the system moves away from the steady state after small perturbations, then the steady state is unstable. The notion of stability is very important to analyze the behavior of dynamical systems. In control theory, it means that if for any finite signal, the system produces a finite output signal, then the system is stable.
For an OTC market, the stability can roughly be understood as a state where, despite small price and volatility changes of assets, the asset
prices are not going far away from the equilibrium prices.
 The fluctuations of prices do not affect roughly the dynamic of markets.
 The asymptotic stability appears as a powerful tool to predict dark markets.
  More precisely it helps to understand that, despite all fluctuations, in the long run, market prices will remain close to steady state prices. More rigorously, let us consider an autonomous differential system $\dot{x}(t)=F(x(t))$  with $F:[0,\infty )\times D\longrightarrow \mathbb{R}^{n}$ a function piecewise continuous in $t$ and locally Lypschitz in $x$, where $D$ is a domain containing the solution $x(t)$.

\begin{definition}
\begin{itemize}

 \item
 A solution $x^{0}$ of the above system is said to be stable if given any  $\epsilon>0$ and any $t_{0}\geq 0$, there exists a $\delta=\delta(\epsilon,t_0)>0$ such that $||x(t_0)-x^{0}(t_0)||<\delta$ implies $||x(t)-x^{0}(t)||< \epsilon,$ for all $t\geq t_{0}\geq 0$, for any solution $x(t).$

 \item

 A solution $x^{0}$ of the above system is said to be asymptotically stable if it is stable and for any $t_{0}\geq 0$, there exists a constant $c=c(t_0)>0$ such that whenever $||x(t_0)-x^{0}(t_0)||<c,$ we have  $\lim\limits_{t \rightarrow \infty} ||x(t)-x^{0}(t)||=0.$

\end{itemize}
\end{definition}

 Given an autonomous differential system $\dot{x}=F(x),$ where $F$ is a twice continuously differentiable function, using Taylor's theorem, we can approximate the model at the steady state with a linear model. A very standard result in dynamical systems states that the original system is asymptotically stable if all the eigenvalues of the approximating linear system have negative real parts (see \cite{Braun1993} p. 386).
Appendix \textbf{A} gives more details and a rigorous treatment of these notions.

In order to study the asymptotic stability of OTC market models with several assets, we need also to recall some concepts related to diagonally dominant matrices. A good presentation of this theory can be found, for instance, in the book \textit{Mathematical Economics} by A. Takayama \cite{Takayama1985}.

Let $A=(a_{ij})$ be an $n\times n$ complex matrix. We say that the matrix $A$ is diagonally dominant (abbreviated d.d.), if there exist positive numbers $d_1,d_2,...,d_n$ satisfying the inequality  $d_{i}|a_{ii}|>\sum \limits_{\underset{j \neq i}{j=1}}^n d_{j}|a_{ij}|$. Technically, this is the definition of row dominance but we could work equivalently with the transpose notion of column dominance. In several works, this is called the GDD (generalized diagonal dominant) matrices and diagonal dominance sometimes means the following stricter property due to Hadamard: for all $i=1,..,n$  $|a_{ii}|>\sum \limits_{\underset{j \neq i}{j=1}}^n |a_{ij}|$. We will call this property strictly d.d. Obviously, $A$ is d.d. if and only if $DA$ is strictly d.d. with $D=diag(d_1,d_2,...,d_n)$. Hadamard shows that strictly d.d. matrices are non-singular. McKenzie provides the proof of its generalization to d.d. matrices. On the other hand, Hadamard's theorem is a direct consequence of Gerschgorin's disks theorem (see, for instance, Varga \cite{Varga2011}.)

\subsection{ Main Result}
We now present the main result of this section in the theorem below.

\begin{theo}
Every partially segmented market model with several assets has a unique steady state which is asymptotically stable.

\end{theo}

It was shown in \cite{Beland2016} that every partially segmented OTC market model with several assets has a steady state which is unique.
For the proof of asymptotic stability, we denote by $x^{\ast}$ the unique equilibrium state, that is $F(x^{\ast})=0.$ We take a multivariate Taylor's expansion of the function $F$ at the equilibrium state $x^{\ast}$ and we just keep a linear approximation of our differential system, that gives rise to the linear system.
The steady state is reached when all of the investors' state proportions remain constant in time, so all derivatives in the left-hand-side of equation (6) are 0. The linear system obtained has the form $\dot{y}=Ay$, where $A$ is the Jacobian matrix at the equilibrium. Note that $y=x-x^{\ast}$ has the origin as equilibrium state, and this equilibrium has the same properties as the steady state $x^{\ast}$.
We have thus reduced the problem to the study of the stability of the matrix $A$, that is to show that all its eigenvalues have a negative real parts. This problem leads us to a new class of matrices that has never been studied before, but we show that it is connected to the diagonally dominant matrices, using the steady state and the results in \cite{Mckenzie2009} and \cite{Garleanu1949}.
Since the square matrix obtained is a matrix with non positive off-diagonal elements, we use the Hawkins-Simon theorem (see \cite{Garleanu1949}) to find a positive vector $d=(d_1,d_2,...,d_3)$ that is a solution of the inequality $Ax> 0$. Finally, we extend MacKenzie's theorem  in the current context.

Let us first simplify the notations of our Master Equations (5). Set $x_{i}=\mu_t(hi,n)$; $x_{K+i}=\mu_t(li,o)$, for all $i=1,2,\cdots, K$ and $m=\sum_{i\in\mathcal{I}} m_i$. The above system with the $1+K$ constraints becomes

\begin{align}\label{eq:pmasterSystem}
\begin{split}
	\dot{x}_{i} &= -\lambda_i x_{i}x_{K+i} - \widetilde{\gamma}_{i}x_{i}- \widetilde{\gamma}_{ui} \sum \limits_{\underset{j \neq i}{j=1}}^n x_{j}+ \widetilde{\gamma}_{ui}(1-m), \ \forall i\in\mathcal{I}\\
	\dot{x}_{K+i} &=  -\lambda_i x_{i}x_{K+i} - \gamma_{i}x_{K+i} + \gamma_{di}m_{i}, \ \forall i\in\mathcal{I}
\end{split}
\end{align}

Let $x=(x_{1},x_{2},...,x_{2K})$ denote a vector in $\mathbb{R}^{2K}$.
And let $I=[0,1]^{2K}$ with $f=(f_1,f_2,...,f_{2K}):I\longrightarrow \mathbb{R}^{2K}$, defined by:

\begin{align}
\begin{split}
 f_{i}(x) &= -\lambda_i x_{i}x_{K+i} - \widetilde{\gamma}_{i}x_{i}- \widetilde{\gamma}_{ui} \sum \limits_{\underset{j \neq i}{j=1}}^n x_{j}+ \widetilde{\gamma}_{ui}(1-m)\\
 f_{K+i}(x) &=-\lambda_i x_{i}x_{K+i} - \gamma_{i}x_{K+i} + \gamma_{di}m_{i}, \ \forall i\in\mathcal{I}.
\end{split}\label{eq:steadyMuhin}
\end{align}

The Jacobian matrix of $f$ is the $2K\times 2K$ matrix given by $J(x)=(\frac{\partial f_i(x)}{\partial x_j})$ where

\begin{description}

\item  $\frac{\partial f_i(x)}{\partial x_j}=-\lambda_i x_{K+i} - \widetilde{\gamma}_{i}$; $\frac{\partial f_i(x)}{\partial x_{K+i}}=-\lambda_i x_{i}$ and $\frac{\partial f_i(x)}{\partial x_{i}}=-\widetilde{\gamma}_{ui}$; $\frac{\partial f_i(x)}{\partial x_{K+j}}=0$ for $j\neq i$
\item  $\frac{\partial f_{K+i}(x)}{\partial x_i}=-\lambda_i x_{K+i}$; $\frac{\partial f_{K+i}(x)}{\partial x_{K+i}}=-\lambda_i x_{i}-\gamma_{i}$ and $\frac{\partial f_{K+i}(x)}{\partial x_{j}}=0$; $\frac{\partial f_{K+i}(x)}{\partial x_{K+j}}=0$

    \item  for $j\neq i$.

\end{description}

Hence the matrix $J(x)$ can be expressed as the block matrix

\[ J(x)= \left[ \begin{array}{c|c} A_{11} & A_{12} \\ \hline A_{21} & A_{22} \end{array} \right] \] where
the matices $A_{11}$, $A_{12}$, $A_{21}$ and $ A_{22}$ are described explicitly in Appendix \textbf{C}.

 Because the matrix $J(x)$ is a square matrix with non positive off-diagonal elements, we use the Hawkins-Simon theorem (see \cite{Garleanu1949}) to find a positive vector $d=(d_1,d_2,...,d_3)$ such that $J(x)d> 0$. This condition is equivalent to showing that all the determinants of the principal minors of this matrix are positive. We notice that, after elementary operations on the lines of these determinants, each of them is equivalent to the well-known Hurwitz determinant. This connection enables us, to show that these determinants are positive. Thus the given matrix is a d.d matrix. Hence the unique steady state $x^{\ast}$ is asymptotically stable, and we have obtained the proof of theorem 3.1.

\bigskip

\section{Equilibrium Inter-investor Asset Pricing}
The steady state equilibrium proportions of investors for a segmented OTC model with multiple assets is determined in \cite{Beland2016} and its asymptotic stability is established above. In this section, we will use the equilibrium masses above, in order to compute the equilibrium bargaining prices $P_1$, $P_2$,...,$P_{K}$ for the respective assets
1,2,...,$K$.
 The authors in \cite{Begm2013} propose a method to show how to obtain the traded prices of the assets at the unique steady state (equilibrium) for non-segmented market models. These prices, however, were not computed explicitly. The objective of this section is to give an explicit formula of the steady state traded prices for any partially segmented market model with multiple assets.

The ownership and price processes, $\theta_{i}$ and $P_i(t)$, respectively, were introduced in section 2 along with the asset value process $dA(t)$. Agents are risk-neutral and infinitely lived. Agents can invest in a bank account with a risk-free interest rate of $r$, assumed to be constant. A low-type agent $li$, when owning the asset has a holding cost $\delta_{di}$ per time unit while a high-type agent has the full holding dividend $\delta_{hi}$. Let $U$ denote a utility function. A $consumption$ $process$ is an $\mathcal{F}_{t}$-progressively
measurable, nonnegative process $c$ satisfying $\displaystyle{\int_{0}^{T}}c(t)~\textrm{d}t < \infty$ almost surely. The $cumulative$ $consumption$ $process$ $C$ is related to the consumption rate process $c$ and is defined by the formula $C(t)=\displaystyle{\int_{0}^{t}}c(s)~\textrm{d}s$, where $0\leq t\leq T$.
Following Duffie et al. \cite{Duffie2005}, we obtain for the case of multiple assets, the wealth equation:
\begin{align}
\begin{split}
  \dd W(t) &= rW(t)\dd t - C(t)\dd t \\
  &\qquad + \sum_{i\in\mathcal{I}}\left[\theta_i(t)\left(\delta_{hi}-\delta_{di}\mathbbm{1}_{\{Z(t)=(li,o)\}}\right)\dd t - P_i(t)\dd \theta_i(t) \right]
\end{split}
\end{align}
where $C$ is a cumulative consumption process, $\theta_i(t) \in \{0,1\}$ is a feasible holding process for the asset $i$,
with initial wealth $W(0) = w_0$.
As usual, prices are obtained from the following expected utility maximization problem:
\begin{align}
  &\sup_{\{C(v),\theta_1(v),...,\theta_K(v)\}} \mathbb{E}\left[\int_t^{\infty}e^{-r(v-t)}U(C(v))\dd v \ | \ Z(t)=z, W(t) = w\right]
\end{align}
subject to the constraint \begin{align}
\begin{split}
  \dd W(t) &= rW(t)\dd t - C(t)\dd t \\
  &\qquad + \sum_{i\in\mathcal{I}}\left[\theta_i(t)\left(\delta_{hi}-\delta_{di}\mathbbm{1}_{\{Z(t)=(li,o)\}}\right)\dd t - P_i(t)\dd \theta_i(t) \right].
\end{split}
\end{align}
Proceeding exactly as in  \cite{Begm2013}, section 4, and in the spirit of Duffie et al. (see \cite{Duffie2005}, \cite{Duffie2007}), the value functions $V(t,z)$ for the states $z\in E$ are linear functions in wealth and satisfy the following Hamilton-Jacobi-Bellman equations:

\begin{eqnarray}
   \dot{V}(t,(l,n))           & =& -\sum_{i\in\mathcal{I}} V(t,(hi,n))\widetilde{\gamma}_{ui} + \left(r + \sum_{i\in\mathcal{I}}\widetilde{\gamma}_{ui} \right) V(t,(l,n))\label{eq:pV(l,n)-point}\\
    \dot{V}(t,(hi,n)) &=&  - \left(V(t,(hi,o))- P_i(t)\right) \lambda_i\mu_t(li,o) - V(t,(l,n)) \widetilde{\gamma}_{di}\label{eq:pV(hi,n)-point} \\
                 & & + \left(\widetilde{\gamma}_{di} + r +\lambda_i \mu_t(li,o)\right) V(t,(hi,n)\nonumber\\
  \dot{V}(t,(hi,o)) &=& \left(\gamma_{di} + r\right)V(t,(hi,o)) - \gamma_{di}V(t,(li,o)) - \delta_{hi} \label{eq:pV(hi,o)-point}\\
  \dot{V}(t,(li,o)) &=& \left(\gamma_{ui} + r + \lambda_i\mu_t(hi,n)\right)V(t,(li,o)) - \gamma_{ui}V(t,(hi,o))\label{eq:pV(li,o)-point} \\
  && -\lambda_i \mu_t(hi,n)(V(t,(l,n))+P_i(t)) - (\delta_{hi} - \delta_{di})\nonumber
\end{eqnarray}

From which we get the steady-state equations:

\begin{align*}
0  &=-\sum_{i\in\mathcal{I}} V(t,(hi,n))\widetilde{\gamma}_{ui} + \left(r +\sum_{i\in\mathcal{I}}\widetilde{\gamma}_{ui}\right) V(t,(l,n))\\
0 &=  - \lambda_i\mu(li,o) \left(V(hi,o)- P_i\right) - \widetilde{\gamma}_{di}  V(l,n) + \left(\widetilde{\gamma}_{di} + r +\lambda_i \mu(li,o)\right) V(hi,n),\\
0 &= \left(\gamma_{di} + r\right)V(hi,o) - \gamma_{di}V(li,o) - \delta_{hi},\\
0 &= \left(\gamma_{ui} + r + \lambda_i\mu(hi,n)\right)V(li,o) - \gamma_{ui}V(hi,o) -\lambda_i \mu(hi,n)(V(l,n)+P_i)\\
 -& (\delta_{hi}-\delta_{di}), \ \forall i\in\mathcal{I}.
\end{align*}

Because we do not consider the presence of market-makers in our markets, the prices are negotiated bilaterally by the agents. A high-type non-owner pays at most his reservation value $\Delta^h_i = V(hi,o)-V(hi,n)$
in order to obtain the asset $i$ and the low-type owner requires a price of at least $\Delta^l_i = V(li,o)-V(l,n)$. The steady state price $P_{i}$ for asset $i$ is such that $\Delta^l_i \leq P_i \leq \Delta^h_i$. Following \cite{Duffie2007}, Nash (1950) bargaining leads to the following equibrium price for asset $i$:
\begin{equation}\label{eq:price}
  P_i = (1-q_{i})\Delta^l_i + q_{i}\Delta^h_i
\end{equation}
where $q_{i} \in (0,1)$ represents the bargaining power of the seller $(li,o)$ and depends only on the asset $i\in\mathcal{I}.$ We assume in the rest of this article that it is constant, that is, $q_{i}=q$ for each asset $i;$ this assumption can be understood as replacing $q_i$ by the mean of the distribution of the bargaining powers of the seller for the above $K$ assets.  We need to solve the above linear system. We set $V(hi,n)=x_{i}$, $V(li,o)=y_{i}$, $V(hi,o)=z_{i}$, $V(l,n)=w$, $\lambda_i\mu(li,0)=a_{i},$
 $\widetilde{\gamma}_{di}+r+\lambda_i\mu(li,0)=b_{i}$, $\gamma_{ui}+r+\lambda_i\mu(hi,n)=b_{i},$ $\lambda_i\mu(hi,n)=d_{i}$ and $r_{i}=\gamma_{ui}+r$. Using these notations, we have $P_i = (1-q)(y_{i}-w) + q(z_{i}-x_{i})$ and hence the above linear system is equivalent to the system

\begin{align*}
0  &=-\sum_{i\in\mathcal{I}} x_{i}\widetilde{\gamma}_{ui} + \left(r +\sum_{i\in\mathcal{I}}\widetilde{\gamma}_{ui}\right) w\\
0 &=  (b_{i}-qa_{i})x_{i} + (1-q)a_{i}y_{i} - (1-q)a_{i}z_{i}- (\widetilde{\gamma}_{di} + (1-q)a_{i})w, \ \forall i\in\mathcal{I}\\
0 &=  r_{i}z_{i} - \gamma_{di}y_{i} - \delta_{hi}, \ \forall i\in\mathcal{I}\\
  0 &= qd_{i}x_{i} + (c_{i}-qd_{i})y_{i} - (\gamma_{ui} + qd_{i})z_{i}- qd_{i}w - (\delta_{hi}-\delta_{di}),\\
  &\qquad
   \ \forall i\in\mathcal{I}.
\end{align*}

It is a linear system and we have three equations for each $i$ in $\mathcal{I}$.
If we let $w$ be the parameter, we obtain a system of three equations in three variables $x_i$, $y_i$ and $z_i$ which are given in terms of $w$. Since  $z_{i} = \frac{\gamma_{di}}{r_{i}}y_{i} + \frac{\delta_{hi}}{r_{i}}$, it is sufficient to express $x_i$ and $y_i$ in term of $w.$ A straightforward calculation gives us

\begin{equation}\label{eq:price}
x_i=\frac{\left(1+\frac{\gamma_i+r}{qd_i}\right) \left(1+\frac{\widetilde{\gamma}_{di}}{(1-q)a_i}\right)-1}{\left(1+\frac{\gamma_i+r}{qd_i}\right)
\left(1+\frac{\widetilde{\gamma}_{di}+r}{(1-q)a_i}\right)-1}w + \frac{\delta_{di}}{\left(1+\frac{\gamma_i+r}{qd_i}\right)\left(1+\frac{\widetilde{\gamma}_{di}+r}{(1-q)a_i}\right)-1}.
 \end{equation}

  For notational simplicity, we set $\Psi(i,r)=\left(1+\frac{\gamma_i+r}{qd_i}\right)\left(1+\frac{\widetilde{\gamma}_{di}}{(1-q)a_i}\right)-1$,  $\Gamma(i,r)=\left(1+\frac{\gamma_i+r}{qd_i}\right)\left(1+\frac{\widetilde{\gamma}_{di}+r}{(1-q)a_i}\right)-1$, $\Lambda(i,r)=\frac{\Psi(i,r)}{\Gamma(i,r)}$  and $\Omega(i,r)=\frac{{\delta}_{di}}{qd_{i}\Gamma(i,r)}$. With these notations we have $x_i=\Lambda(i,r)w+\Omega(i,r).$ Now we can compute $w$ explicitly using the expression of $x_i$ and the first equation of our linear system. Thus $w=\frac{\sum_{i\in\mathcal{I}}\widetilde{\gamma}_{ui}\Omega(i,r)}{r +\sum_{i\in\mathcal{I}}\widetilde{\gamma}_{ui}(1-\Lambda(i,r))}$ which is positive because $\Lambda(i,r)$ belongs to $(0,1)$. We obtain explicit expressions of $x_i$ and $y_i$, which, in turn, give us the formula for the price $P_i$ as shown below.

\begin{theo}
For any partially segmented market model, there exists a unique steady-state equilibrium and the equilibrium prices are given by
\begin{align}
\begin{split}
P_{i} &= \left( \frac{{\gamma}_{di}}{(1-q)a_{i}\Gamma(i,r)}-\Lambda(i,r)\right) \frac{\sum_{i\in\mathcal{I}}\widetilde{\gamma}_{ui}\Omega(i,r)}{r +\sum_{i\in\mathcal{I}}\widetilde{\gamma}_{ui}(1-\Lambda(i,r))}+
\frac{{\delta}_{hi}}{r}\\
&\qquad -q\Omega(i,r)\left(1+\left(1-r+\frac{{\gamma}_{di}}{q}\right)\left(1+\frac{\widetilde{\gamma}_{di}}{(1-q)a_i}\right)\right), \ \forall i\in\mathcal{I}\end{split}\label{eq:steadyMuhin}
\end{align}

\end{theo}

For each asset $i$, two bargainers, a seller and a buyer, are faced with a finite number of possible alternatives. If the two investors agree on an alternative payoff, the transaction is executed. Otherwise, if they fail to agree, the transaction does not occur. When the bargaining mechanism begins, the seller of the asset $i$ proposes the price $P^{+}_{i}$ and accepts any price greater than or equal to $P^{+}_{i}$, the buyer in turn proposes the price $P^{-}_{i}$ and accepts all prices less than or equal to $P^{-}_{i}$. Here we are in presence of a Nash bargaining problem or the Rubinstein-type game, for more details we refer to (\cite{Duffie2007},\cite{Osborne1990}). A $Nash$ $equilibrium$ is an action profile with the property that, no single player can obtain a higher payoff by deviating unilaterally from this profile; a player will not gain by changing his strategy. The unique price $P^{\ast}_{i}$ at which the transaction is executed is a Nash equilibrium. It is well-known that there exists a vector of payoffs $(P^{\ast}_1,P^{\ast}_2,...,P^{\ast}_{K})$ which is the unique Nash bargaining vector solution.
The negotiation process stops when it reachs a Nash equilibrium.
Following (\cite{Duffie2007},\cite{Osborne1990}) we apply the devise of Rubinstein-Wolinsky to calculate explicitly the unique bargaining powers that represent the prices. For a given asset $i$, the two investors find each other randomly, the seller with probability $\hat{q}$ and the buyer with probability $1-\hat{q},$ to suggest a trading price. The seller of the asset $i$ makes an offer and the buyer who needs this asset proposes another price. The seller can accept or reject the offer. If the offer is rejected, the seller of the asset $i$ receives the dividend from the asset $i$ during this period. The bargaining may break down before a counter offer is made. At the next period $\Delta^{i}t$ later, one of the investor is chosen to make a new offer. This process occurs simultaneously for each of the trading assets $1,2,...,K.$ The next period at which the first investor in the market is chosen randomly after leaving his first partner is $\min \{\Delta^{i}t, \, 1\leq i\leq K\}$; and the next period at which the last investor in the market is chosen randomly after leaving his first partner is $\max \{\Delta^{i}t, \, 1\leq i\leq K\}.$ The limiting price as $\Delta^{i}t$ approaches zero is represented by $P_i = (1-q)\Delta^l_i + q\Delta^h_i$, with the bargaining power of the seller $q$ equal to $\hat{q}.$
An impatient agent who wants to optimize its profit can search for alternative partners during negotiations to get the best offer for the asset $i$. If he finds a godsend or he trusts to find one, he leaves his partner for a new one. In this market, for each asset $j\neq i$, a similar mechanism may occur.
But there is no interaction between two agents who are bargaining for two separate assets, because the market is partially segmented. If an agent leaves his partner, he will bargain with another investor that has the same focus regarding asset. However, an unsophisticated investor does not search for an alternative partner during the bargaining, in this case the formulae of payoff price remains unchanged except the new bargaining power that becomes $q=\frac{1}{K}\sum \limits_{\underset{}{i=1}}^K   \frac{\hat{q}(r+\gamma_i+\widetilde{\gamma}_i+\lambda_{i}\mu(li,o))}
{\hat{q}(r+\gamma_i+\widetilde{\gamma}_i+\lambda_{i}\mu(li,o))+(1-\hat{q})(r+\gamma_i+\widetilde{\gamma}_i+\lambda_{i}\mu(hi,n))}$.
We solve the homogeneous linear system obtained from the equations $(9)-(12)$ to get the equilibrium prices as the unique solution of this system. Thus we have the prices $P_1,P_2,...,P_{K}.$ \newline
If we agree that the potential number of open days for which the stock market is open in a year is around 250, given the equilibrium proportion of potential buyers, $\mu(hi,n)$, for the asset $i$, the average number of days needed to sell the asset $i$ is given by $250(\lambda_i\mu(hi,n))^{-1}$, with $i=1,2,...,K.$ This average time is called $equilibrium$ $timing$ $of$ $seller$. We obtained a finite sequence of times for all sold
assets. The last asset to be sold gives the maximum  equilibrium timing and the first asset sold gives the minimum equilibrium timing.


\section{A Numerical Example}
 For this numerical illustration, we consider an OTC market with $K=2$ assets. The search intensity of $\lambda_1=1250$ in the first table below means that an investor looking for asset 1 expects to be in contact with 1250 investors each year, that is on average $1250\diagup 250=5$ investors per day. The search intensity of $\lambda_2=2000$ in the first table below means that an investor looking for asset 2 expects to be in contact with 2000 investors each year, that is $2000\diagup 250=8$ investors per day on average.
The parameters for the pricing model can be found in the table below.

\vspace{0.5cm}

\begin{tabular}{|*{12}{c|}}
    \hline
     $\lambda_1$  & $\lambda_2$  & $\gamma_{u1}$ & $\gamma_{d1}$ & $\gamma_{u2}$ & $\gamma_{d2}$  & $m_{1}$  & $m_{2}$  & $\widetilde{\gamma}_{u1}$  & $\widetilde{\gamma}_{d1}$  & $\widetilde{\gamma}_{u2}$  & $\widetilde{\gamma}_{d2}$  \\
    \hline
     1250  & 2000  & 5  & 0.5  & 8 & 3 & 0.3 & 0.6 & 2.5 & 3.5 & 0.4 & 1.5   \\
    \hline

\end{tabular}

\vspace{0.3cm}
The second group of parameters for the pricing model is given in the following table
\begin{center}
\begin{tabular}{|*{6}{c|}}
    \hline
   $\delta_{h1}$  & $\delta_{h2}$  & $\delta_{d1}$  & $\delta_{d2}$ & q & r \\
    \hline
     2.5 & 3.5 & 0.4 & 1.5 & 0.5 & 0.05  \\
    \hline

\end{tabular}
\end{center}
\vspace{0.3cm}
Now we can compute the unique equilibrium vector proportions and asset prices at the equilibrium.
We recall that the equilibrium occurs when we have
\begin{align}
0 &=-\lambda_i \mu(hi,n)\mu(li,o) + \widetilde{\gamma}_{ui} \mu(l,n) - \widetilde{\gamma}_{di}\mu(hi,n), \  i\in \{1,2\} \label{eq:psolveEq1} \\
0 &= -\lambda_i \mu(hi,n)\mu(li,o) - \gamma_{ui}\mu(li,o) + \gamma_{di}\mu(hi,o), \  i\in \{1,2\}\label{eq:psolveEq2}
\end{align}
with the constraints
\begin{align}
	\mu(hi,o) + \mu(li,o) &= m_i, \  i\in \{1,2\} \label{eq:psolveConstr1} \\
	 m_1+ m_2 + \mu(h1,n) + \mu(h2,n) + \mu(l,n) &= 1 \label{eq:psolveConstr2}
\end{align}

The equations $(20)$ and $(21)$ imply for $i\neq j$, the following system
\begin{align*}
  \mu(hi,n)  &=  - \frac{\widetilde{\gamma}_{ui}}{\widetilde{\gamma}_i} \mu(hj,n) + \frac{%
\gamma_{i}\gamma_{di}m_{i}}{\widetilde{\gamma}_{i}(\lambda_i
\mu(hi,n)+\gamma_{i})} + \widetilde{\gamma}_{ui}\left(1 - m \right)  + \frac{\gamma_{di}}{%
\widetilde{\gamma} _{i}}m_{i}, \ i\in \{1,2\}
.\end{align*}
By setting $x=\mu(h1,n)$ and $y=\mu(h2,n)$ and replacing all parameters by their numerical values
we get the following system of two equations with two variables.


\begin{equation*}
  \left\{
    \begin{aligned}
      y&=- 1.05 x + \frac{0.1375}{1250
x+5.5} + 0.175&\text{(1)}\\
      x&=- 1.2 y + \frac{1.98}{2000
y+11} + 0.02&\text{(2)}\\
    \end{aligned}
  \right.
\end{equation*}

We can substitute $y$ in the second equation with its value in the first equation. By straightforward calculations, one obtains
the following quartic equation $2294300x^{4}+606494x^{3}-75548.2132x^{2}-693.572x-1.500884=0.$ The only positive root of this equation is
$x^{\ast}=0.0991$, that is $\mu^{\ast}_{(h1,n)}=0.0991$ and therefore $\mu^{\ast}_{(h2,n)}=0.0720$. From the relation
\begin{align}
\mu(li,o) &= \frac{\gamma_{di} m_i}{\lambda_i \mu(hi,n) + \gamma_i}, \ i\in \{1,2\}
\end{align}
 at the equilibrium, we have
$\mu^{\ast}_{(l1,0)}=0.0011$ and $\mu^{\ast}_{(l2,0)}=0.0116$; moreover, we use the constraints $(22)$ and $(23)$ to get
$\mu^{\ast}_{(h1,0)}=0.2984$, $\mu^{\ast}_{(h2,o)}=0.5883$  and $\mu^{\ast}_{(l,n)}=0.0289.$ It remains to compute the equilibrium prices
$P_1$ and $P_2$ for the two assets. From Theorem 4.1, we have

\begin{align}
\begin{split}
P_{i} &= \left( \frac{{\gamma}_{di}}{(1-q)a_{i}\Gamma(i,r)}-\Lambda(i,r)\right) \frac{\sum_{i\in\mathcal{I}}\widetilde{\gamma}_{ui}\Omega(i,r)}{r +\sum_{i\in\mathcal{I}}\widetilde{\gamma}_{ui}(1-\Lambda(i,r))}+
\frac{{\delta}_{hi}}{r}\\
&\qquad -q\Omega(i,r)\left(1+(1-r+\frac{{\gamma}_{di}}{q})(1+\frac{\widetilde{\gamma}_{di}}{(1-q)a_i})\right), \ i\in \{1,2\} \end{split}\label{eq:steadyMuhin}
 .\end{align}
Before computing the values $P_1$ and $P_2$ we need to determine the intermediate settings that appear in the formulae. They have been grouped in the following table

\vspace{0.3cm}

\begin{center}
\begin{tabular}{|*{8}{c|}}
    \hline
      $\Gamma(1,r)$ & $\Gamma(2,r)$  & $\Omega(1,r)$  & $\Omega(2,r)$  & $\Lambda(1,r)$  & $\Lambda(2,r)$  & $\Psi(1,r)$ & $\Psi(2,r)$ \\
    \hline
      5.7159 & 0.3075 & 0.0011 & 0.0677 & 0.9861 & 0.9837 & 5.6366 & 0.3026  \\
    \hline

\end{tabular}
\end{center}

\vspace{0.3cm}
Using the data in the table above and the formula $(25)$, we have $P_1=50.0031$ and $P_2=69.6551.$
The following table gives the equilibrium proportions and asset prices for the model with two assets:


\vspace{0.3cm}

\begin{tabular}{|p{1cm}|*{8}{@{\hskip.9mm}c@{\hskip.9mm}|}}
\hline
      $\mu^{\ast}_{(h1,n)}$ & $\mu^{\ast}_{(h2,n)}$ & $\mu^{\ast}_{(l1,o)}$ & $\mu^{\ast}_{(l2,o)}$ & $\mu^{\ast}_{(h1,o)}$ &    $\mu^{\ast}_{(h2,o)}$  & $\mu^{\ast}_{(l,n)}$ & $P_{1}$  & $P_{2}$ \\ \hline
     0.0991 & 0.0720 & 0.0011 & 0.0116 & 0.2989 & 0.5883 & 0.0289 & 50.0031 & 69.6551\\ \hline

\end{tabular}

\vspace{0.3cm}
 We recall that if the equilibrium proportion of potential buyers for the asset $i$ is known, the average time needed to sell the asset $i$ is $250(\lambda_i\mu(hi,n))^{-1}$, with $i=1,2;$ that is $250(\lambda_1\mu(h1,n))^{-1}=2$ days to sell asset 1 and $250(\lambda_2\mu(h2,n))^{-1}=1.7$ days to sell asset 2.

\section{Comparative Statics and Cross Effect}

We would like to study the rate of change of the equilibrium prices when some exogenous parameters of the model vary simultaneously. Because these models allow multiple assets, we cannot adapt the technique of Duffie \cite{Duffie2012} which is based on the analysis of each exogenous parameter independently.
For the comparative statics of our model, the equilibrium prices are regarded as functions with several variables and we use, for a given asset $i$, the limiting bargaining power associated during negotiation. When we change the same type of parameters, the equilibrium prices can still exist or the equilibrium can be broken. It is important to vary the same type of parameters at once and it is useful to determine the behavior of the quantities $\Psi(i,r)$,  $\Gamma(i,r)$, $\Lambda(i,r)$  and $\Omega(i,r).$ The following result characterizes the variation of equilibrium prices when the same type of parameters approach infinity. One can see that, if all $\gamma_{ui}$ and $\gamma_{di}$ approach infinity, then $\Omega(i,r)$ vanishes while $\Psi(i,r)$ and $\Gamma(i,r)$ approach infinity. We also have that $\Psi(i,r)$ and  $\Gamma(i,r)$ become closer to zero as $\lambda_i$ tends to infinity for all $i$. Then for the same variation of $\lambda_i$, the expression $\Omega(i,r)$ converges to $\widehat{\Omega}(i,r)=\frac{{\delta}_{di}}{\frac{q}{1-q}\frac{\mu(hi,n)}{\mu(li,o)}(\widetilde{\gamma}_{di}+r)+ \gamma_i+r}$, while $1-\Lambda(i,r)$ converges to $1-\widehat{\Lambda}(i,r)=\frac{rq\mu(hi,n)}{\delta_{di}\mu(li,o)}\widehat{\Omega}(i,r)$ and $\frac{{\gamma}_{di}}{(1-q)a_{i}\Gamma(i,r)}$ tends to $\frac{q\gamma_{di}\mu(hi,n)}{(1-q)\delta_{di}\mu(li,o)}\widehat{\Omega}(i,r).$


\begin{prop}
The equilibrium prices satisfy the following properties.

\begin{itemize}

 \item When $\gamma_{uj}$ approaches infinity for all $j$, the price $P_i$ converges to \newline $\widehat{P}_i(\lambda)(\gamma_{u})=\frac{\delta_{hi}}{r}$;
\item When for all $j$, $\gamma_{dj}$ approaches infinity, the price $P_i$ converges to \\ $\overline{P}_i(\lambda)(\gamma_{u})=\frac{\delta_{hi}}{r}-\delta_{di}\frac{(1-q)\lambda_{i}\mu(hi,n)+\widetilde{\gamma}_{di}}{(1-q)
    \lambda_{i}\mu(li,o)+\widetilde{\gamma}_{di}+r}$;
\item If for all $j$, $\widetilde{\gamma}_{uj}$ approaches infinity or $\widetilde{\gamma}_{dj}$ approaches infinity, then the limit of the price $P_i$ does not exist;
\item When $\lambda_j$ approaches infinity for all $j$, the price $P_i$ converges to $\widehat{P}_i$, where the exact expression is given by

  \begin{align}
\begin{split}
\widehat{P}_i &= \left( \frac{q\gamma_{di}\mu(hi,n)}{(1-q)\delta_{di}\mu(li,o)}\widehat{\Omega}(i,r)-\widehat{\Lambda}(i,r)\right) \frac{\sum_{i\in\mathcal{I}}\widetilde{\gamma}_{ui}\widehat{\Omega}(i,r)}{r +\sum_{i\in\mathcal{I}}\widetilde{\gamma}_{ui}(1-\widehat{\Lambda}(i,r))}\\+
\frac{{\delta}_{hi}}{r}
&\qquad -\left( \frac{{\delta}_{di}(2-r+\frac{{\gamma}_{di}}{q})}{\frac{q}{1-q}\frac{\mu(hi,n)}{\mu(li,o)}(\widetilde{\gamma}_{di}+r)+ \gamma_i+r}\right), \ \forall i\in\mathcal{I}.\end{split}\label{eq:steadyMuhin}
\end{align}

\end{itemize}

\end{prop}

We observe that there are two types of exogenous parameters, the homogeneous parameters whose variation maintains the existence of equilibrium prices and the heterogeneous parameters that express the segmentation which break it. By changing randomly these parameters, the equilibrium prices can be broken. It means that we cannot go from a non-segmented model to a partially segmented model by a continuum of values. Which implies that, in every partially segmented market, the knowledge of the homogeneous parameters are not sufficient to capture optimal information and percolation in the market. Thus, to understand the behavior of any partially segmented market, we need the whole parameters and the unique equilibrium steady state. There are also two sorts of homogeneous parameters affecting the behavior of the equilibrium prices; the limit prices when both $\lambda_i$ and $\gamma_{ui}$ tend to infinity is the same when the parameters $\gamma_{ui}$ tend to infinity. The same phenomenon happens if we replace $\gamma_{di}$ by $\gamma_{ui}$ in the above limit. This framework shows that the behavior of the equilibrium prices at infinity are captured by the only parameters $\gamma_{ui}$ or $\gamma_{di}.$

   Now we analyse the cross-price elasticity that is, the responsiveness of the variation of an exogenous parameter for the asset $i$ to change in the price of another asset, $ceteris$ $paribus$. We investigate how the effect that an exogenous parameter attached to the asset $i$ can have on the price $P_j$, where $j\neq i$. For example, if we consider the price $P_j$ with regards to variations of the exogenous parameter $\lambda_{i}$, we denote $P_j=P_j(\lambda_{i}).$ Surprisingly, there is a constant value depending on the $j^{th}$ parameter of the model where $\lambda_{i}$ vanishes in the expression $P_j=P_j(\lambda_{i}).$
We are now in the position to prove the following result which provides a nice characterization of the equilibrium effect on price sensitivities.

\begin{prop}

There exists a constant $\widehat{\lambda}(j)$ depending on the exogenous parameters of the model and the equilibrium steady state such that:
\begin{enumerate}
\item the equilibrium price $P_j(\lambda_{i})$ is increasing when $\widehat{\lambda}(j)< 0$;
\item the equilibrium price $P_j(\lambda_{i})$, with $j\neq i$, is constant when $\lambda_{j}$ is equal to $\widehat{\lambda}(j)>0$, that is $P_j(\lambda_{i})$ is invariant in $\lambda_i$;
\item the equilibrium price $P_j(\lambda_{i})$ are increasing for $\lambda_{i}> \widehat{\lambda}(j)$ and are decreasing for $0< \lambda_j \leq \widehat{\lambda}(j).$
\end{enumerate}




\end{prop}

The new parameter $\widehat{\lambda}(j)$ depends on the exogenous parameters of the model.
Since the transactions in OTC markets are done online or by phone, the network effect is almost inevitable. The parameters that measure the network are the frequencies $\lambda_i$.
For every asset with liquid payoff, it will produce the network effect, that is, many people will converge on the bargain and thus increase the use of the asset. The facility to access the network can decrease or increase the bargaining power of some investors. The investors who are far geographically, such as in certain areas of Africa or Asia, and who do not have sophisticated materials and high-quality networks are disadvantaged compared to others located for instance in New York. The growth of the market and the presence of multiple assets with liquidity payoffs can cause the $network$ $congestion$.

For illustration, we consider a market with two assets. We present in Figure 1 the behavior of the price of asset 1 as a function of $\lambda_2$ which is the encounter frequency for exchanging item 2, for fixed values of $\lambda_1.$ We notice that $P_1$ is a decreasing function of $\lambda_2$. When the frequency of meetings to exchange asset 1 is sufficiently large ( $\lambda_1$=6250) the decreasing of $P_1$, when $\lambda_2$ increases, accelerates. In general, when $\lambda_1$ and $\lambda_2$ tend to infinity, the price of each asset tends to a value that can be explicitly calculated using Proposition 7.1, thus $\widehat{P}_1=49.0824$ and $\widehat{P}_2=57.2058$. Another thing to note is that, although $P_1$ is a monotone function of $\lambda_2$, it is not a monotone function of $\lambda_1.$ However, for sufficiently large values of $\lambda_1$, $P_1$ becomes a monotone function of $\lambda_1$, which explains the fact that $P_1$ converges asymptotically when the intensities of meetings tend to infinity. A lesson to be kept is that the availability of information reduces prices if the quantity of information is sufficiently large. On the other hand, the intensities of meetings for asset 1 can be increasing, but with $P_1$ being decreasing. In Figure 1, for example, we have ($5> 0.002$) but the corresponding values of the price $P_1$ are decreasing. This observation allows us to argue that the intensities of meetings are not the only parameters that influence the equilibrium prices.
\noindent

Figure 2 illustrates the behavior of the price of asset 1 as a function of $\lambda_1,$ which is the encounter frequency for exchanging item 1, for fixed values of $\lambda_2.$ We notice that $P_1$ is an increasing function of $\lambda_1$ in this particular example. In general, the price $P_1$ is not a monotone function of $\lambda_1.$ Its monotonicity depends on the sign of  $\frac{q\gamma_{d1}\mu(h1,n)}{(1-q)\delta_{d1}\mu(l1,o)}\Omega(1,r)-\Lambda(1,r),$  which itself depends on $\lambda_1.$ When the intensities of meetings to exchange asset 2  are sufficiently large ($\lambda_2$=250, $\lambda_2$=1250, $\lambda_2$=6250) the increase of $P_1$, when $\lambda_1$ increases, decelerates. In this case, prices are very close and seem to converge towards a single price. Surprisingly, we note that when the intensities of meetings to exchange asset 2  are sufficiently small ($\lambda_2$=0.002), prices are also close to previous ones.  But for a random intensity of meetings to exchange asset 2, ($\lambda_2$=5), the price $P_1$ has a high growth rate and is greater than the previous prices. So, we do not have a monotonic behavior of equilibrium prices as a function of the frequency of meetings, but rather a chaotic price movement over the intensities of meetings.

\begin{figure}
\centering
\begin{tikzpicture}
\begin{axis}[width=\textwidth, 
height=0.67\textwidth, 
axis x line*=bottom, 
axis y line*=left, 
grid=none, 
xmin=0, 
xmax=100, 
every axis x label/.style={
    at={(ticklabel* cs:1.05)},
    anchor=west, 
},
every axis y label/.style={
    at={(ticklabel* cs:1.05)},
    anchor=south, 
},
xlabel=$\lambda_2$, 
ylabel=$P_1(\lambda_2)$, 
legend pos=south west] 

\addplot[thick] [domain=0:100, samples=100]{(-16.4316*x-0.0017)/(52.254*x+8030.8289)+49.9307};
\addplot[blue,thick] [domain=0:100, samples=100]{(-16.3878*x-10.7315)/(52.254*x+8030.8289)+49.8757};
\addplot[red,thick] [domain=0:100, samples=100]{(-15.8534*x-145.9938)/(52.1104*x+8010.3218)+49.97};
\addplot[green,thick] [domain=0:100, samples=100]{(-15.8534*x-193.7486)/(52.1104*x+8010.3218)+49.9927};
\addplot[yellow,thick] [domain=0:100, samples=100]{(-73.1477*x-978.2355)/(51.9668*x+7989.8148)+49.997};

\legend{$\lambda_1=0.002$,$\lambda_1=5$, $\lambda_1=250$, $\lambda_1=1250$, $\lambda_1=6250$};

\end{axis}

\end{tikzpicture}
\caption{Variation of equilibrium prices $P_1(\lambda_2)$}

\end{figure}
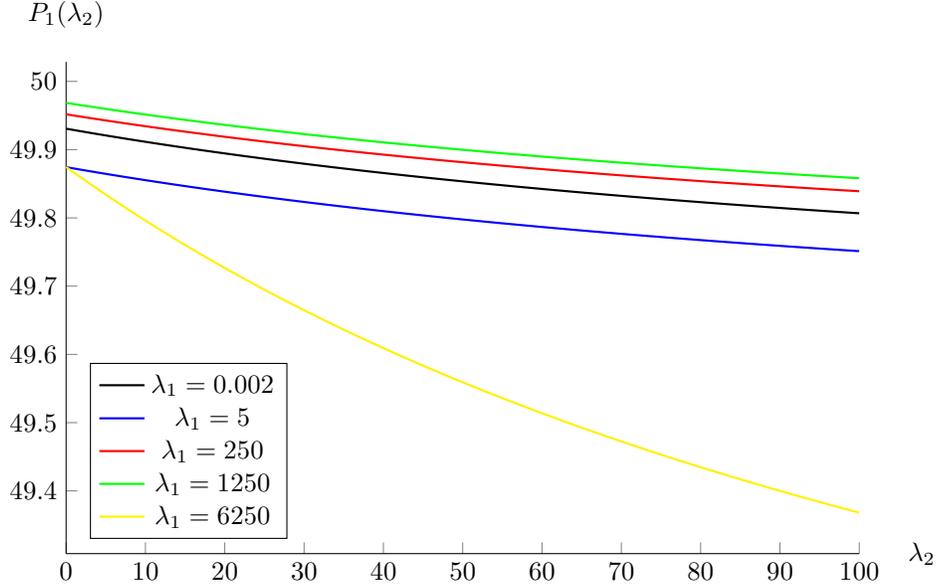




\begin{figure}
\centering
\begin{tikzpicture}
\begin{axis}[width=\textwidth, 
height=0.67\textwidth, 
axis x line*=bottom, 
axis y line*=left, 
grid=none, 
xmin=0, 
xmax=100, 
every axis x label/.style={
    at={(ticklabel* cs:1.05)},
    anchor=west, 
},
every axis y label/.style={
    at={(ticklabel* cs:1.05)},
    anchor=south, 
},
xlabel=$\lambda_1$, 
ylabel=$P_1(\lambda_1)$, 
legend pos=south east] 

\addplot[blue,thick] [domain=0:100, samples=100]{(-112341.3974*x^2-14384932.7307*x)/((6566.5534*x+722960.7211)*(639.6522*x+70916.1363))+
(-11.907*x-50086.8123)/(6566.5534*x+722960.7211) + 50};
\addplot[thick,red] [domain=0:100, samples=100]{(-112341.3974*x^2-14384937.7406*x-641.5003)/((6566.5534*x+722960.7211)*(4394.3389*x+70713.7073))+
(-11.907*x-50086.8123)/(6566.5534*x+722960.7211) + 50};
\addplot[thick,green] [domain=0:100, samples=100]{(-112341.3974*x^2-14385035.1553*x-13115.1184)/((6566.5534*x+722960.7211)*(610.7594*x+67735.1091))+
(-11.907*x-50086.8123)/(6566.5534*x+722960.7211) + 50};
\addplot[thick, dash pattern=on 4pt off 1pt on 4pt off 4pt] [domain=0:100, samples=100]{(-112341.3974*x^2-14384934.0667*x-171.0667)/((6566.5534*x+722960.7211)*(598.9396*x+66433.7798))+
(-11.907*x-50086.8123)/(6566.5534*x+722960.7211) + 50};

 \addplot[thick,brown]
 [domain=0:100, samples=100]{(-112341.3974*x^2-14385090.5981*x-20171.6223)/((6566.5534*x+722960.7211)*(595.6563*x+66028.9218))+
(-11.907*x-50086.8123)/(6566.5534*x+722960.7211) + 50};

\legend{$\lambda_2=0.002$,$\lambda_2=5$, $\lambda_2=250$, $\lambda_2=1250$, $\lambda_2=6250$};

\end{axis}
\end{tikzpicture}
\caption{Variation of equilibrium prices $P_1(\lambda_1)$}
\end{figure}
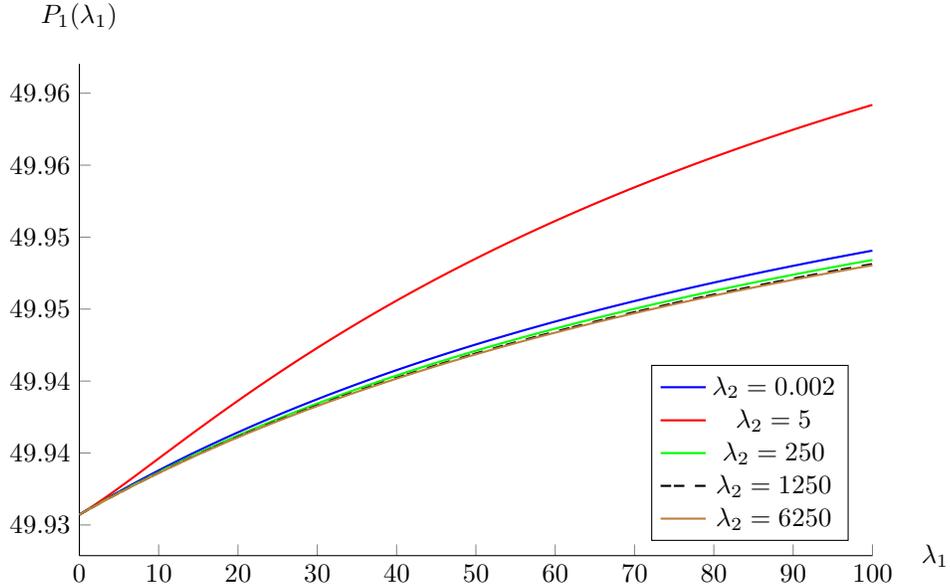

\newpage

\section{Conclusion}

This paper generalizes the model of OTC market due to  Duffie-G\^{a}rleanu-Pedersen \cite{Duffie2005}. Our market is segmented and involves multiple assets, and we analyze its long-run prediction. We show that the
 unique steady state of any segmented OTC market is asymptotically stable. In addition, we find a closed form solution for the equilibrium price of each asset. Our analysis shows how prices respond to the level of interactions between agents and the other exogenous parameters of the model. We also obtain the rate of change of the equilibrium prices when some exogenous parameters of the model grow simultaneously to infinity. In addition to the substantive contributions, we contribute methodologically by developing new tools to study segmented dark markets with multiple assets.


\vspace{1cm}

\section*{Acknowledgements}
The first author is supported in part by a team grant from Fonds de Recherche du Qu\'ebec - Nature et Technologies (FRQNT grant no. 180362).

\newpage




\begin{thebibliography}{99}




\bibitem{Begm2013} B\'elanger, A., Giroux, G. and Moisan-Poisson, M.:
\emph{Over-the-counter market models with several assets}. q-fin.CP, arxiv:1308.2957v1

\bibitem{Beland2016} B\'elanger, A. and Ndoun\'e, N.:
 \emph{ Existence and Uniqueness of a Steady State for an OTC  Market with Several Assets}. Math. and Finan. Econ. (1), 1815-1847 (2016)

\bibitem{Braun1993} Braun, M.:
\emph{Differential Equations and Their Applications}. Pringer-Verlag. Edit. 4, (1993)





\bibitem{Duffie2001} Duffie, D.:
\emph{Dynamic Asset Pricing Theory}. Princeto.University. Press. Edit. 3. (2001)

\bibitem{Duffie2012} Duffie, D.:
\emph{Dark Markets: Asset pricing and information transmission in over-the-counter
markets}. Princeto.University. Press. Edit. 2. (2012)

\bibitem{Duffie2005} Duffie, D., G\^arleanu, N., Pedersen, L.H.:
 \emph{Over-the-Counter
Markets}. Econometrica. 73(1), 1815-1847 (2005)

\bibitem{Duffie2007} Duffie, D., G\^arleanu, N., Pedersen, L.H.:
\emph{Valuations in Over-the-Counter
Markets}. Rev.Financia. Stud. 20, 1865-1900 (2007)


\bibitem{Duffie2014} Duffie, D., Malamud, S., Manso, G.:
 \emph{Information percolation in segmented markets}. Forthcoming, Journ. Econ. Theo.


\bibitem{Duffiesun2012} Duffie, D., Sun, Y.:
\emph{The exact law of large numbers for independent
random matching}. Journ. Econ. Theo. 1105-1139 (2012)






\bibitem{Garleanu1949} Hawkins, D., Simon, H.A.:
\emph{Some Conditions of Macroeconomic Stability}. Econometrica. 144(3), 245-248 (1949)


\bibitem{Mckenzie2009} Mckenzie, L.:
\emph{ Matrices with Dominant Diagonals and Economic Theory, in K. J.
Arrow, S. Karlin, and P. Suppes (eds.)}. Mathematical Methods in the Social Sciences. Stanford
University Press, 47-62  (1959).


\bibitem{Platen2006} Platen, E., Heath, D.:
\emph{A Benchmark Approach to Quantitative Finance}. Mathematical Finan. 16(1), 131-151 (2006)


\bibitem{Osborne1990} Osborne, M.J., Rubinstein A.:
\emph{Bargaining markets}. Acad. Press.  (1990)













\bibitem{Sun2006} Sun, Y.:
\emph{The exact law of large numbers via fubini extension and characteri-
zation of insurable risks}. Journ. Econ. Theo. 126, 31-69 (2006)


\bibitem{Takayama1985} Takayama, A.:
\emph{Mathematical {E}conomics}. Cambridge University Press. (1985)



\bibitem{Varga2011} Varga, R.S.:
\emph{Gerschgorin and his Circles}. Cambridge University Press. (4) (1998)

\bibitem{Vayanos2007} Vayanos, D., Wang, T.:
\emph{Search and Endogeneous Concentration of Liquidity in Asset Markets}. Journ. Econ. Theo. 166(11), 66-104 (2007)

\bibitem{Weill2008} Weill, P.O.:
\emph{Liquidity premia in dynamic bargaining markets}. Journ. Econ. Theo. 140, 66-96 (2008)

\end{thebibliography}


\appendix

\section{Asymptotic Stability for ODE systems}\label{app}

We consider an autonomous differential system $$\dot{x}(t)=F(x(t)),$$ with $F:[0,\infty )\times D\longrightarrow \mathbb{R}^{n}$ piecewise continuous in $t$ and locally Lypschitz in $x$, where $D$ is a domain containing the solution $x^{0}$.

The linear part of $F$ at the point $x^{0}$, denoted $A=Df(x^{0})$ is the Jacobian matrix of $F$ at the point $x^{0}.$ Since the function is differentiable and its differential is continuous, Taylor's Theorem for functions of several variables says that $f(x)=Df(x^{0})(x-x^{0})+ \theta(x)$, where $\theta$ is the function that is small near $x^{0}$ in the sense that $\lim\limits_{x \rightarrow x^{0}} \frac{\theta(x)}{||x-x^{0}||}=0.$
For a differential equation $\dot{x}(t)=F(x(t))$, as before, let $x^{\ast}$ denote a steady state and let $A=Df(x^{\ast})$ denote the Jacobian matrix of $F$ at the point $x^{\ast}.$ Let $(\lambda_{i})$ be the family of all eigenvalues of $A$. We have the following stability theorems (See \cite{Braun1993}, p.378).

\begin{theo}
A steady state equilibrium $x^{\ast}$ of the above system is stable if and only if both of the following conditions hold:

\begin{itemize}
\item all eigenvalues of $A$ have non positive real part, that is $\mathfrak{R}(\lambda_{i})\leq 0.$

\item for every eigenvalue of real part $\mathfrak{R}(\lambda_{i})= 0$ and algebraic multiplicity $\eta_i\geq 1,$ we have
 Rank$(A-\lambda_{i}I)=\eta_{i}$

\end{itemize}
\end{theo}
We are in position to give a very important result which characterizes the asymptotic stability.

\begin{theo}
The steady state $x^{\ast}$ of the above system is asymptotically stable if and only if
all eigenvalues of $A$ have a negative real part, that is $\mathfrak{R}(\lambda_{i})< 0,$ for all $i$.

\end{theo}

\section{ Diagonally dominant matrices}\label{app}

\begin{theo} (Hadamard-Levy-Desplanques)
If an $n\times n$ matrix $A$ is d.d. then it is non-singular.
\end{theo}
This theorem is a a particular case of the Gerschgorin theorem.
We then have the closely related result.
\begin{theo} ( McKenzie)
If an $n\times n$ matrix $A$ has a generalized diagonal dominant that is negative, then all its eigenvalues have negative real parts.
\end{theo}

As we will see, the problem of showing that a matrix is d.d. can oftentimes be reduced to finding a positive solution to a linear system, it will be useful to recall the Hawkins and Simon condition.

\begin{theo} (Hawkins-Simon)
Let $A$ be an $n\times n$ matrix with non positive off-diagonal elements. Then the following statements are equivalent:
\begin{description}

\item(i) the inequality $AX > 0$ has a positive solution,
\item(ii) $det\left(
\begin{array}{ccc}
   a_{11} & \cdots & a_{1p} \\
   \vdots & \ddots & \vdots \\
   a_{p1} & \cdots & a_{pp}
\end{array}
\right)
> 0$ for all $p=1,2,\cdots,n$.

\end{description}
\end{theo}

\section{ Proof of the results}\label{app}

\begin{lem}
The matrix $J(x)$ is a d.d. matrix.
\end{lem}

\begin{proof}
We recall that the matrix $J(x)$ can be expressed as the block matrix
\[ J(x)= \left[ \begin{array}{c|c} A_{11} & A_{12} \\ \hline A_{21} & A_{22} \end{array} \right] \], where
$$A_{11}=\left(
\begin{array}{ccccc}
-\lambda_{1}x_{K+1}-\widetilde{\gamma}_{1} & -\widetilde{\gamma}_{u1}   &     &\cdots&  -\widetilde{\gamma}_{u1} \\
-\widetilde{\gamma}_{u2} & -\lambda_{2}x_{K+2}-\widetilde{\gamma}_{2}   &  &\vdots& \\
 &\ddots&\ddots&\ddots&  \\
\vdots&  &  -\widetilde{\gamma}_{u_{K-1}}  &     &   -\widetilde{\gamma}_{u_{K-1}} \\
 -\widetilde{\gamma}_{u_{K}} &  \cdots&  &   -\widetilde{\gamma}_{u_{K}}  &  -\lambda_{K}x_{2K}-\widetilde{\gamma}_{K}
\end{array}
\right)
$$

$$A_{12}=\left(
\begin{array}{ccccc}
-\lambda_{1}x_{1} & 0   &     &\cdots&  0 \\
0 & -\lambda_{2}x_{2}   &  &\vdots& \\
 &\ddots&\ddots&\ddots&  \\
\vdots&  &  0  &     &   0 \\
 0 &  \cdots&  &   0  &  -\lambda_{K}x_{K}
\end{array}
\right)
$$

$$A_{21}=\left(
\begin{array}{ccccc}
-\lambda_{1}x_{K+1} & 0   &     &\cdots&  0 \\
0 & -\lambda_{2}x_{K+2}   &  &\vdots& \\
 &\ddots&\ddots&\ddots&  \\
\vdots&  &  0  &     &   0 \\
 0 &  \cdots&  &   0  &  -\lambda_{K}x_{2K}
\end{array}
\right)
$$
and

$$A_{22}=\left(
\begin{array}{ccccc}
-\lambda_{1}x_{1}-\gamma_{1} & 0   &     &\cdots&  0 \\
0 & -\lambda_{2}x_{2}-\gamma_{2}   &  &\vdots& \\
 &\ddots&\ddots&\ddots&  \\
\vdots&  &  0  &     &   0 \\
 0 &  \cdots&  &   0  &  -\lambda_{K}x_{K}-\gamma_{1}
\end{array}
\right)
$$

  Note that the matrix $J(x)$ satisfies the condition $|a_{ii}|>\sum \limits_{\underset{j \neq i}{j=1}}^n |a_{ij}|$ for $i\in \{K,K+1,\cdots,2K\}$.
We need to find a sequence of real numbers $(d_{i})_{1\leq i\leq n}$, with $d_{i}> 0$ for each $i$ such that $DJ$ becomes a stricly d.d.  matrix, where $D$ is the diagonal matrix given by $$D=\left(
\begin{array}{ccccc}
d_{1} & 0   &     &\cdots&  0 \\
0 & d_{2}   &  &\vdots& \\
 &\ddots&\ddots&\ddots&  \\
\vdots&  &  0  &     &   0 \\
 0 &  \cdots&  &   0  &  d_{2K}
\end{array}
\right)
.$$ The matrices $A_{12},A_{21}$ and $A_{22}$ are diagonal matrices with negative
entries.

The matrix $DJ$ is strictly d.d. if we have

\[ (E)
\left \{
\begin{array}{c @{>} c}
    d_{i}(\lambda_{i}x_{K+i}+\widetilde{\gamma}_{i}) & \sum \limits_{\underset{j \neq i}{j=1}}^n d_{j}\widetilde{\gamma}_{ui}+d_{K+i}\lambda_{i}x_{K+i} \\
    d_{K+i}(\lambda_{i}x_{i}+\gamma_{i}) & d_{i}\lambda_{i}x_{i} \\
\end{array}
\right.
\]

Which can be rewritten:

\[ \left \{
\begin{array}{c @{>} c}
    d_{i}(\lambda_{i}x_{K+i}+\widetilde{\gamma}_{i}) & \sum \limits_{\underset{j \neq i}{j=1}}^n d_{j}\widetilde{\gamma}_{ui}+d_{K+i}\lambda_{i}x_{K+i} \\
    d_{K+i} &\frac{d_{i}\lambda_{i}x_{i}}{\lambda_{i}x_{i}+\gamma_{i}}  \\
\end{array}
\right.\]

 Thus the system of $2K$ inequalities with $2K$ unknown variables is reduced to the following system of $K$ inequalities with $K$ variables
 $$(\star) \; d_{i}(\lambda_{i}x_{K+i}+\widetilde{\gamma}_{i}) > \sum \limits_{\underset{j \neq i}{j=1}}^n d_{j}\widetilde{\gamma}_{ui}+
    \frac{d_{i}\lambda_{i}^{2}x_{i}x_{K+i}}{\lambda_{i}x_{i}+\gamma_{i}}.$$ Indeed, when we get the positive $d_i$ for $i=1,...,K$ we let $d_{K+i}$ be such that $0<d_{K+i}-\frac{d_{i}\lambda_{i}x_{i}}{\lambda_{i}x_{i}+\gamma_{i}}<\frac{\epsilon}{2\lambda_{i}x_{i}}$ where $\epsilon =  d_{i}(\lambda_{i}x_{K+i}+\widetilde{\gamma}_{i}) - \sum \limits_{\underset{j \neq i}{j=1}}^n d_{j}\widetilde{\gamma}_{ui}-
        \frac{d_{i}\lambda_{i}^{2}x_{i}x_{K+i}}{\lambda_{i}x_{i}+\gamma_{i}}.$

Then we get $
    d_{i}\left(\lambda_{i}x_{K+i}+\widetilde{\gamma}_{i}-\frac{\lambda_{i}^{2}x_{i}x_{K+i}}{\lambda_{i}x_{i}+\gamma_{i}}\right) > \sum \limits_{\underset{j \neq i}{j=1}}^n d_{j}\widetilde{\gamma}_{ui}.$

    Because $x$ is the steady state, we get from (11) that $x_{K+i}=\frac{\gamma_{di}m_{i}}{\lambda_{i}x_{i}+\gamma_{i}},$ and we have

\begin{eqnarray*}\lambda_{i}x_{K+i}+\widetilde{\gamma}_{i}-\frac{\lambda_{i}^{2}x_{i}x_{K+i}}{\lambda_{i}x_{i}+\gamma_{i}}
&=&\lambda_{i}x_{K+i}+\widetilde{\gamma}_{i}-\frac{\lambda_{i}^{2}x_{i}x_{K+i}^{2}}{\gamma_{di}m_{i}}
\\
&=&\lambda_{i}x_{K+i}+\widetilde{\gamma}_{i}-\frac{\lambda_{i}^{2}x_{K+i}^{2}}{\gamma_{di}m_{i}}\left(\frac{\gamma_{di}m_{i}}{\lambda_{i}x_{K+i}}
-\frac{\gamma_{i}}{\lambda_{i}}\right)\\
&=&\widetilde{\gamma}_{i}+\frac{\lambda_{i}\gamma_{i}}{\gamma_{di}m_{i}}x_{K+i}^{2}.\end{eqnarray*}
We define the $K \times K$ matrix $B=(b_{ij})$, where $b_{ii}=\widetilde{\gamma}_{i}+\frac{\lambda_{i}\gamma_{i}}{\gamma_{di}m_{i}}x_{K+i}^{2}$ and $b_{ij}=-\widetilde{\gamma}_{i}$ for $j \neq i$.  Clearly $b_{ii}$ is a positive quantity for all $i.$
The system of inequality $(\star)$ is equivalent to the system
$(\star \star)$ $d_{i}a_{ii} > \sum \limits_{\underset{j \neq i}{j=1}}^n d_{j}\widetilde{\gamma}_{ui}.$
If $D_K$ is a diagonal matrix with diagonal entries $d_1,d_2,...,d_K$, we may write it as $D_K=diag(d_1,d_2,...,d_K)$, then we want $BD_K> 0.$ From Theorem Appendix B.3 (Hawkins-Simons), the inequality $BD_K>0$, with $B$ having non positive off-diagonal elements has a positive solution if and only if
$det\left(
\begin{array}{ccc}
   b_{11} & \cdots & b_{1p} \\
   \vdots & \ddots & \vdots \\
   b_{p1} & \cdots & b_{pp}
\end{array}
\right)
> 0$ for all $p=1,2,\cdots,K$.

Let $\tilde{B}_{pp}=\left(
\begin{array}{ccc}
   b_{11} & \cdots & b_{1p} \\
   \vdots & \ddots & \vdots \\
   b_{p1} & \cdots & b_{pp}
\end{array}
\right)$
We will now compute the determinant of the matrix $\tilde{B}_{pp}$.
$det(\tilde{B}_{pp})=\frac{1}{\prod \limits_{\underset{}{j=1}}^p \widetilde{\gamma}_{ui}}\left|
\begin{array}{ccccc}
\frac{b_{11}}{\widetilde{\gamma}_{u1}} & -1   &     &\cdots&  -1 \\
-1 & \frac{b_{22}}{\widetilde{\gamma}_{u2}}   &  &\vdots& \\
 &\ddots&\ddots&\ddots&  \\
\vdots&  &  -1  &     &   -1 \\
 -1 &  \cdots&  &   -1  &  \frac{b_{pp}}{\widetilde{\gamma}_{up}}
\end{array}
\right|$

\begin{eqnarray*}det(\tilde{B}_{pp})
&=&\frac{1}{\prod \limits_{\underset{}{j=1}}^p \widetilde{\gamma}_{uj}}\prod \limits_{\underset{}{j=1}}^p \left (1+\frac{b_{jj}}{\widetilde{\gamma}_{uj}}\right) \left(1-\sum \limits_{\underset{}{j=1}}^p\frac{(-1)}{1+\frac{b_{jj}}{\widetilde{\gamma}_{uj}}}\right)
\\&=&\prod \limits_{\underset{}{j=1}}^p \left(\frac{1+\frac{b_{jj}}{\widetilde{\gamma}_{uj}}}{\widetilde{\gamma}_{uj}}\right) \left(1+\sum \limits_{\underset{}{j=1}}^p\frac{1}{1+\frac{b_{jj}}{\widetilde{\gamma}_{uj}}}\right),\end{eqnarray*}

The second member of the product in the first line above is the Hurwitz determinant. We do have that $det(\tilde{B}_{pp})>0.$
Since $det(\tilde{B}_{pp})>0$ and $B$ has non positive off-diagonal elements, then there exists a positive vector $(d_1,d_2,...,d_K)$ in $\mathbb{R}^{K}$ such that $B (d_1,d_2,...,d_K)>0.$ Hence there exists a diagonal matrix
$D=diag(d_{1},...,d_{2k})$ such that $DA$ is a strictly d.d. matrix.

Finally, from Theorem 3.1, the steady state of a partially segmented OTC market model is asymptotically stable if the Jacobian matrix evaluated at the steady state is such that all its eigenvalues have negative real parts. We have just shown that the matrix $J(x)$ is d.d. and its diagonal is negative. An appeal to McKenzie's Theorem (Appendix B.2) gives us the desired conclusion that all eigenvalues of $J(x)$ have negative real parts, hence $x$ is stable. This completes the proof of our main result.

\end{proof}

\begin{proof} \,\textbf{Theorem 4.1}

Since we know that $P_i = (1-q)\Delta^l_i + q\Delta^h_i$, it is sufficient to compute each term of this equality.

We recall the following notations: $V(hi,n)=x_{i}$, $V(li,o)=y_{i}$, $V(hi,o)=z_{i}$, $V(l,n)=w$, $\lambda_i\mu(li,0)=a_{i},$
 $\widetilde{\gamma}_{di}+r+\lambda_i\mu(li,0)=b_{i}$, $\gamma_{ui}+r+\lambda_i\mu(hi,n)=b_{i},$ $\lambda_i\mu(hi,n)=d_{i}$ and $r_{i}=\gamma_{ui}+r$. Using these notations, we have $P_i = (1-q)(y_{i}-w) + q(z_{i}-x_{i}).$

From equality $(15)$ and the expression of $w$ we have

\begin{equation}\label{eq:price}
x_i=\frac{\sum_{i\in\mathcal{I}}\widetilde{\gamma}_{ui}\Omega(i,r)}{r +\sum_{i\in\mathcal{I}}\widetilde{\gamma}_{ui}(1-\Lambda(i,r))}\Lambda(i,r)+ \Omega(i,r)
 \end{equation}
and

\begin{align}
\begin{split}
y_{i} &= \left( \frac{r_{i}}{(1-q)a_{i}\Gamma(i,r)}\right) \frac{\sum_{i\in\mathcal{I}}\widetilde{\gamma}_{ui}\Omega(i,r)}{r +\sum_{i\in\mathcal{I}}\widetilde{\gamma}_{ui}(1-\Lambda(i,r))}+
\frac{{\delta}_{hi}}{r}\\
&\qquad -\Omega(i,r)\left(1+\frac{\widetilde{\gamma}_{di}}{(1-q)a_i}\right)\frac{r_{i}}{r}. \end{split}\label{eq:steadyMuhin}
\end{align}
We also have

\begin{align}
\begin{split}
 \Delta^l_i &= \left( \frac{r_{i}}{(1-q)a_{i}\Gamma(i,r)}-1\right) \frac{\sum_{i\in\mathcal{I}}\widetilde{\gamma}_{ui}\Omega(i,r)}{r +\sum_{i\in\mathcal{I}}\widetilde{\gamma}_{ui}(1-\Lambda(i,r))}+
\frac{{\delta}_{hi}}{r}\\
&\qquad -\Omega(i,r)\left(1+\frac{\widetilde{\gamma}_{di}}{(1-q)a_i}\right)\frac{r_{i}}{r} \end{split}\label{eq:steadyMuhin}
\end{align}

and

\begin{align}
\begin{split}
\Delta^h_i &= \left( \frac{{\gamma}_{di}}{(1-q)a_{i}\Gamma(i,r)}-\Lambda(i,r)\right) \frac{\sum_{i\in\mathcal{I}}\widetilde{\gamma}_{ui}\Omega(i,r)}{r +\sum_{i\in\mathcal{I}}\widetilde{\gamma}_{ui}(1-\Lambda(i,r))}+
\frac{{\delta}_{hi}}{r}\\
&\qquad - \Omega(i,r)\left(1+\frac{\widetilde{\gamma}_{di}}{(1-q)a_i}\right)\frac{\gamma_{di}}{r}-\Omega(i,r). \end{split}\label{eq:steadyMuhin}
\end{align}

A straightforward calculation and the relation $\Gamma(i,r)=\Psi(i,r)+\frac{r}{(1-q)a_{i}}$ give the result.

\end{proof}

\begin{proof} \,\textbf{Proposition 6.1}

We recall that the prices are given by
\begin{align}
\begin{split}
P_{i} &= \left( \frac{{\gamma}_{di}}{(1-q)a_{i}\Gamma(i,r)}-\Lambda(i,r)\right) \frac{\sum_{i\in\mathcal{I}}\widetilde{\gamma}_{ui}\Omega(i,r)}{r +\sum_{i\in\mathcal{I}}\widetilde{\gamma}_{ui}(1-\Lambda(i,r))}+
\frac{{\delta}_{hi}}{r}\\
&\qquad -q\Omega(i,r)\left(1+(1-r+\frac{{\gamma}_{di}}{q})(1+\frac{\widetilde{\gamma}_{di}}{(1-q)a_i})\right), \ \forall i\in\mathcal{I}.\end{split}
\end{align}
 As before set $\Psi(i,r)=(1+\frac{\gamma_i+r}{qd_i})(1+\frac{\widetilde{\gamma}_{di}}{(1-q)a_i})-1$,  $\Gamma(i,r)=(1+\frac{\gamma_i+r}{qd_i})(1+\frac{\widetilde{\gamma}_{di}+r}{(1-q)a_i})-1$, $\Lambda(i,r)=\frac{\Psi(i,r)}{\Gamma(i,r)}$, $\Omega(i,r)=\frac{{\delta}_{di}}{qd_{i}\Gamma(i,r)}$ and $\Upsilon(i,r)=q\Omega(i,r)\left(1+(1-r+\frac{{\gamma}_{di}}{q})(1+\frac{\widetilde{\gamma}_{di}}{(1-q)a_i})\right)$.\\
 We have $\displaystyle\lim_{\gamma_{ui}\to \infty}\Gamma(i,r)=\infty$ and
 $\displaystyle\lim_{\gamma_{ui}\to \infty}\Omega(i,r)=\displaystyle\lim_{\gamma_{ui}\to \infty}\frac{\delta_{di}}{qd_{i}\Gamma(i,r)}=0$. Moreover we have
\begin{eqnarray*}\displaystyle\lim_{\gamma_{ui}\to \infty}\Lambda(i,r)&=&\frac{1+\frac{\widetilde{\gamma}_{di}}{(1-q)a_i}}{1+\frac{\widetilde{\gamma}_{di}+r}{(1-q)a_i}}\\&=
&\frac{(1-q)\lambda_{i}\mu_(li,o)+\widetilde{\gamma}_{di}}{(1-q)\lambda_{i}\mu_(li,o)+\widetilde{\gamma}_{di}+r}.
\end{eqnarray*}
These expressions imply that $\displaystyle\lim_{\gamma_{ui}\to \infty}\frac{\sum_{i\in\mathcal{I}}\widetilde{\gamma}_{ui}\Omega(i,r)}{r +\sum_{i\in\mathcal{I}}\widetilde{\gamma}_{ui}(1-\Lambda(i,r))}=0$. Since the expression $1+(1-r+\frac{{\gamma}_{di}}{q})(1+\frac{\widetilde{\gamma}_{di}}{(1-q)a_i})$ is not a function of $\gamma_{ui}$, the above calculations and the algebra of limits allow us to conclude that \newline $\displaystyle\lim_{\gamma_{ui}\to \infty}P_{i}=\frac{{\delta}_{hi}}{r}$.\\\\

 Similarly, we have $\displaystyle\lim_{\gamma_{di}\to \infty}\Gamma(i,r)=\infty$ and $\displaystyle\lim_{\gamma_{ui}\to \infty}\Omega(i,r)=0.$
 Moreover we have
\begin{eqnarray*}\displaystyle\lim_{\gamma_{di}\to \infty}\Lambda(i,r)&=&\frac{\frac{1}{qd_{i}}(1+\frac{\widetilde{\gamma}_{di}}{(1-q)a_i})}{\frac{1}{qd_{i}}(1+\frac{\widetilde{\gamma}_{di}+r}
{(1-q)a_i})}\\&=
&\frac{(1-q)\lambda_{i}\mu_(li,o)+\widetilde{\gamma}_{di}}{(1-q)\lambda_{i}\mu_(li,o)+\widetilde{\gamma}_{di}+r}.
\end{eqnarray*}
Previous calculations imply that $\displaystyle\lim_{\gamma_{di}\to \infty}\frac{\sum_{i\in\mathcal{I}}\widetilde{\gamma}_{ui}\Omega(i,r)}{r +\sum_{i\in\mathcal{I}}\widetilde{\gamma}_{ui}(1-\Lambda(i,r))}=0$. Furthermore, we have
\begin{eqnarray*}\displaystyle\lim_{\gamma_{di}\to \infty}\Upsilon(i,r)&=&\lim_{\gamma_{di}\to \infty} q\Omega(i,r)\left(1+(1-r+\frac{{\gamma}_{di}}{q})(1+\frac{\widetilde{\gamma}_{di}}{(1-q)a_i})\right)\\&=
&\displaystyle\lim_{\gamma_{di}\to \infty} \frac{\delta_{di}\left(1+(1-r+\frac{{\gamma}_{di}}{q})(1+\frac{\widetilde{\gamma}_{di}}{(1-q)a_i})\right)}{d_{i}\left((1+\frac{\gamma_i+r}{qd_i})
(1+\frac{\widetilde{\gamma}_{di}+r}{(1-q)a_i})-1\right)}\\&=
&\frac{\delta_{di}\left((1-q)\lambda_{i}\mu_(li,o)+\widetilde{\gamma}_{di}\right)}{(1-q)\lambda_{i}\mu_(li,o)+\widetilde{\gamma}_{di}+r}.
\end{eqnarray*} Hence the limit of $P_i$ when $\gamma_{di}$ tends to infinity is $$\displaystyle\lim_{\gamma_{di}\to \infty}P_i=\frac{{\delta}_{hi}}{r}-\frac{\delta_{di}\left((1-q)\lambda_{i}\mu_(li,o)+\widetilde{\gamma}_{di}\right)}{(1-q)\lambda_{i}
\mu_(li,o)+\widetilde{\gamma}_{di}+r}$$\\\\

Here, if we assume that $\widetilde{\gamma}_{di}$ runs to the infinity, we have $\displaystyle\lim_{\gamma_{di}\to \infty}\Gamma(i,r)=\infty$ and $\displaystyle\lim_{\gamma_{ui}\to \infty}\Omega(i,r)=0.$ In addition, we have
\begin{eqnarray*}\displaystyle\lim_{\gamma_{di}\to \infty}\Lambda(i,r)&=&\frac{\frac{1}{(1-q)a_{i}}(1+\frac{\widetilde{\gamma}_{ai}}{qd_i})}{\frac{1}{(1-q)a_{i}}(1+\frac{\widetilde{\gamma}_{di}+r}
{qa_i})}\\&=
&1.
\end{eqnarray*}

\begin{eqnarray*}\displaystyle\lim_{\gamma_{di}\to \infty}\Upsilon(i,r)&=&\lim_{\gamma_{di}\to \infty} q\Omega(i,r)\left(1+(1-r+\frac{{\gamma}_{di}}{q})(1+\frac{\widetilde{\gamma}_{di}}{(1-q)a_i})\right)\\&=
&\displaystyle\lim_{\gamma_{di}\to \infty} \frac{\delta_{di}\left(1+(1-r+\frac{{\gamma}_{di}}{q})(1+\frac{\widetilde{\gamma}_{di}}{(1-q)a_i})\right)}{d_{i}\left((1+\frac{\gamma_i+r}{qd_i})
(1+\frac{\widetilde{\gamma}_{di}+r}{(1-q)a_i})-1\right)}\\&=
&\frac{\delta_{di}\left(q(1-r)+\gamma_{di}\right)}{q\lambda_{i}\mu_(hi,n)+ \gamma_{i}+r}.
\end{eqnarray*}

 Thus the limit of $P_i$ when $\widetilde{\gamma}_{di}$ tends to infinity is $$\displaystyle\lim_{\gamma_{di}\to \infty}P_i=\frac{{\delta}_{hi}}{r}-\frac{\delta_{di}\left(q(1-r)+\gamma_{di}\right)}{q\lambda_{i}\mu_(hi,n)+ \gamma_{i}+r}$$\\\\

 We investigate now what happens when $\widetilde{\gamma}_{ui}$ tends to infinity. Because the parameter $\widetilde{\gamma}_{ui}$ appears just in the expression $\Theta(i,r)=\frac{\sum_{i\in\mathcal{I}}\widetilde{\gamma}_{ui}\Omega(i,r)}{r +\sum_{i\in\mathcal{I}}\widetilde{\gamma}_{ui}(1-\Lambda(i,r))}$, it is sufficient to show that the limit of $\Theta(i,r)$ when $\widetilde{\gamma}_{ui}$ tends to infinity does not exist.
  Assume in the first case that $\widetilde{\gamma}_{ui}=\widetilde{\gamma}_{uj}$ for all $i,j$ in $\{1,2,...,K\}$, then

\begin{eqnarray*}\displaystyle\lim_{\widetilde{\gamma}_{ui}\to \infty}\Theta(i,r)&=&\lim_{\widetilde{\gamma}_{ui}\to \infty}\frac{\sum_{i\in\mathcal{I}}\widetilde{\gamma}_{ui}\Omega(i,r)}{r +\sum_{i\in\mathcal{I}}\widetilde{\gamma}_{ui}(1-\Lambda(i,r))}\\&=
&\frac{\sum_{i\in\mathcal{I}}\Omega(i,r)}{r +\sum_{i\in\mathcal{I}}(1-\Lambda(i,r))}.
\end{eqnarray*}
We assume in the second case that $\widetilde{\gamma}_{ui}=i\widetilde{\gamma}_{u1}$ for all $i$ in $\{1,2,...,K\}$, then
\begin{eqnarray*}\displaystyle\lim_{\widetilde{\gamma}_{ui}\to \infty}\Theta(i,r)&=&\lim_{\widetilde{\gamma}_{ui}\to \infty}\frac{\sum_{i\in\mathcal{I}}\widetilde{\gamma}_{ui}\Omega(i,r)}{r +\sum_{i\in\mathcal{I}}\widetilde{\gamma}_{ui}(1-\Lambda(i,r))}\\&=
&\frac{\sum_{i\in\mathcal{I}}i\Omega(i,r)}{r +\sum_{i\in\mathcal{I}}i(1-\Lambda(i,r))}.
\end{eqnarray*}
Since the two expressions $\frac{\sum_{i\in\mathcal{I}}\Omega(i,r)}{r +\sum_{i\in\mathcal{I}}(1-\Lambda(i,r))}$ and
$\frac{\sum_{i\in\mathcal{I}}i\Omega(i,r)}{r +\sum_{i\in\mathcal{I}}i(1-\Lambda(i,r))}$ are not equal, this implies that the price $P_i$ has two different limits; this is a contradiction in virtue of the uniqueness of the limit. Hence the limit $\displaystyle\lim_{\widetilde{\gamma}_{ui}\to \infty}P_{i}$ does not exist.\\\\

Now we want to compute the limit: $\displaystyle\lim_{\lambda_{i}\to \infty}P_i$.  We have
$\displaystyle\lim_{\lambda_{i}\to \infty}\Gamma(i,r)=0$ and $\displaystyle\lim_{\lambda_{i}\to \infty}\Psi(i,r)=0$.

 Furthermore, we have \begin{eqnarray*}\displaystyle\lim_{\lambda_{i}\to \infty}\Lambda(i,r)&=&\lim_{\gamma_{di}\to \infty}\frac{(1+\frac{\gamma_i+r}{qd_i})(1+\frac{\widetilde{\gamma}_{di}}{(1-q)a_i})-1}{(1+\frac{\gamma_i+r}{qd_i})
(1+\frac{\widetilde{\gamma}_{di}+r}{(1-q)a_i})-1}\\&=
&\frac{\frac{\gamma_i+r}{q\lambda_i\mu_(hi,n)}+\frac{\widetilde{\gamma}_{di}}{(1-q)\lambda_i\mu_(li,o)}}{\frac{\gamma_i+r}
{q\lambda_i\mu_(hi,n)}+\frac{\widetilde{\gamma}_{di}+r}{(1-q)\lambda_i\mu_(li,o)}}\\&=
&\frac{(1-q)(\gamma_{i}+r)\mu_(li,o)+q\widetilde{\gamma}_{di}\mu_(hi,n)}{(1-q)(\gamma_{i}+r)
\mu_(li,o)+q(\widetilde{\gamma}_{di}+r)\mu_(hi,n)}.
\end{eqnarray*}

And we deduce the expression \begin{eqnarray*}\displaystyle\lim_{\lambda_{i}\to \infty}(1-\Lambda(i,r))&=&1-\frac{(1-q)(\gamma_{i}+r)\mu_(li,o)+q\widetilde{\gamma}_{di}\mu_(hi,n)}{(1-q)(\gamma_{i}+r)
\mu_(li,o)+q(\widetilde{\gamma}_{di}+r)\mu_(hi,n)}\\&=&\frac{qr\mu_(hi,n)}{(1-q)(\gamma_{i}+r)
\mu_(li,o)+q(\widetilde{\gamma}_{di}+r)\mu_(hi,n)}.
\end{eqnarray*}

We compute also these following limits, \begin{eqnarray*}\displaystyle\lim_{\lambda_{i}\to \infty}\Omega(i,r)&=&\displaystyle\lim_{\lambda_{i}\to \infty} \frac{{\delta}_{di}}{qd_{i}\Gamma(i,r)}\\&=&\frac{qr\mu_(hi,n)}{(1-q)(\gamma_{i}+r)
\mu_(li,o)+q(\widetilde{\gamma}_{di}+r)\mu_(hi,n)};
\end{eqnarray*}

\begin{eqnarray*}\displaystyle\lim_{\lambda_{i}\to \infty}\frac{{\gamma}_{di}}{(1-q)a_{i}\Gamma(i,r)}&=&\frac{q\gamma_{di}\mu_(hi,n)}{(1-q)\delta_{di}\mu_(li,o}
\widehat{\Omega}(i,r)\\&=&\frac{q\gamma_{di}\mu_(hi,n)}{(1-q)(\gamma_{i}+r)
\mu_(li,o)+q(\widetilde{\gamma}_{di}+r)\mu_(hi,n)}
\end{eqnarray*}

and

$$\displaystyle\lim_{\lambda_{i}\to \infty}\left( \frac{{\gamma}_{di}}{(1-q)a_{i}\Gamma(i,r)}-\Lambda(i,r)\right)=\frac{q\gamma_{di}\mu_(hi,n)-(1-q)(\gamma_{i}+r)
\mu_(li,o)-q\widetilde{\gamma}_{di}\mu_(hi,n)}{(1-q)(\gamma_{i}+r)
\mu_(li,o)+q(\widetilde{\gamma}_{di}+r)\mu_(hi,n)}.$$
Replacing each limit computed above by its value allows us to obtain the following expression of prices when $\lambda_{i}\to \infty$:

 \begin{align}
\begin{split}
P_{i} &= \left( \frac{q\gamma_{di}\mu_(hi,n)}{(1-q)\delta_{di}\mu_(li,o}\widehat{\Omega}(i,r)-\widehat{\Lambda}(i,r)\right) \frac{\sum_{i\in\mathcal{I}}\widetilde{\gamma}_{ui}\widehat{\Omega}(i,r)}{r +\sum_{i\in\mathcal{I}}\widetilde{\gamma}_{ui}(1-\widehat{\Lambda}(i,r))}+
\frac{{\delta}_{hi}}{r}\\
&\qquad -\left( \frac{{\delta}_{di}(2-r+\frac{{\gamma}_{di}}{q})}{\frac{q}{1-q}\frac{\mu_(hi,n)}{\mu_(li,o}(\widetilde{\gamma}_{di}+r)+ \gamma_i+r}\right), \ \forall i\in\mathcal{I}.\end{split}
\end{align}

\end{proof}

\begin{proof} \,\textbf{Proposition 6.2}
The prices $P_i$ as functions of $\lambda_{j}$, with $j\neq i$ are given by \begin{align}
\begin{split}
P_{i}(\lambda_{j}) &= \left( \frac{{\gamma}_{di}}{(1-q)a_{i}\Gamma(i,r)}-\Lambda(i,r)\right) \frac{\sum_{i\in\mathcal{I}}\widetilde{\gamma}_{ui}\Omega(i,r)}{r +\sum_{i\in\mathcal{I}}\widetilde{\gamma}_{ui}(1-\Lambda(i,r))}+
\frac{{\delta}_{hi}}{r}\\
&\qquad -q\Omega(i,r)\left(1+(1-r+\frac{{\gamma}_{di}}{q})(1+\frac{\widetilde{\gamma}_{di}}{(1-q)a_i})\right), \ \forall i\neq j\end{split}
.\end{align}

We need to write $P_{i}(\lambda_{j})$ as an explicit function of $\lambda_{j}$ and after that, we will compute the derivatives to show whether $P_{i}(\lambda_{j})$ is increasing or decreasing.

Let us recall that $\Theta(\lambda_{1},...,\lambda_{K})=\frac{\sum_{i\in\mathcal{I}}\widetilde{\gamma}_{ui}\Omega(i,r)}{r +\sum_{i\in\mathcal{I}}\widetilde{\gamma}_{ui}(1-\Lambda(i,r))}.$ The only expression of $P_i$ where an occurrence of $\lambda_{j}$ appears is $\Theta(\lambda_{j}),$ the others expressions do not depend on $\lambda_{j}.$
We must first compute explicitly the expressions $\Omega(i,r)$ and $1-\Lambda(i,r).$
\begin{eqnarray*}\Omega(i,r)&=&\frac{\delta_{di}}{qa_{i}\Gamma(i,r)
}\\&=
&\frac{\delta_{di}\lambda_{i}}{q\mu_(hi,n)\left( \frac{\gamma_i+r}
{q\mu_(hi,n)}\lambda_i+\frac{\widetilde{\gamma}_{di}+r}{(1-q)\mu_(li,o)}\lambda_i+\frac{\gamma_i+r}
{q\mu_(hi,n)}\frac{\widetilde{\gamma}_{di}+r}{(1-q)\mu_t(li,o)}\right)}\\&=
&\frac{\delta_{di}\lambda_{i}}{(\gamma_{i}+r+\frac{\widetilde{\gamma}_{di}+r}{(1-q)\mu_(li,o)}q\mu_(hi,n))\lambda_{i}
+(\gamma_{i}+r)\frac{\widetilde{\gamma}_{di}+r}{(1-q)\mu_(li,o)}}
\end{eqnarray*}
and
\begin{eqnarray*}\Lambda(i,r)&=&\frac{(1+\frac{\gamma_i+r}{qd_i})(1+\frac{\widetilde{\gamma}_{di}}{(1-q)a_i})-1}{(1+\frac{\gamma_i+r}{qd_i})
(1+\frac{\widetilde{\gamma}_{di}+r}{(1-q)a_i})-1}\\&=
&\frac{\frac{\gamma_i+r}
{q\mu_(hi,n)}\lambda_i+\frac{\widetilde{\gamma}_{di}}{(1-q)\mu_(li,o)}\lambda_i+\frac{\gamma_i+r}
{q\mu_(hi,n)}\frac{\widetilde{\gamma}_{di}}{(1-q)\mu_(li,o)}}{\frac{\gamma_i+r}
{q\mu_(hi,n)}\lambda_i+\frac{\widetilde{\gamma}_{di}+r}{(1-q)\mu_(li,o)}\lambda_i+\frac{\gamma_i+r}
{q\mu_(hi,n)}\frac{\widetilde{\gamma}_{di}+r}{(1-q)\mu_(li,o)}}.
\end{eqnarray*}

We can deduce the expression \begin{eqnarray*}1-\Lambda(i,r)&=&1-\frac{\frac{\gamma_i+r}
{q\mu_(hi,n)}\lambda_i+\frac{\widetilde{\gamma}_{di}}{(1-q)\mu_(li,o)}\lambda_i+\frac{\gamma_i+r}
{q\mu_(hi,n)}\frac{\widetilde{\gamma}_{di}}{(1-q)\mu_(li,o)}}{\frac{\gamma_i+r}
{q\mu_(hi,n)}\lambda_i+\frac{\widetilde{\gamma}_{di}+r}{(1-q)\mu_(li,o)}\lambda_i+\frac{\gamma_i+r}
{q\mu_(hi,n)}\frac{\widetilde{\gamma}_{di}+r}{(1-q)\mu_(li,o)}}\\&=&\frac{r(q\lambda_{i}\mu_(hi,n)+\gamma_i+r)}
{(1-q)\mu_(li,o)\left((\gamma_{i}+r+\frac{\widetilde{\gamma}_{di}+r}{(1-q)\mu_(li,o)})\lambda_{i}
+(\gamma_{i}+r)\frac{\widetilde{\gamma}_{di}+r}{(1-q)\mu_(li,o)}\right)}.
\end{eqnarray*}

By setting $a=\gamma_{i}+r+\frac{\widetilde{\gamma}_{di}+r}{(1-q)\mu_(li,o)}q\mu_(hi,n)
$, $b=(\gamma_{i}+r)\frac{\widetilde{\gamma}_{di}+r}{(1-q)\mu_(li,o)}$,\\ $b_1=\frac{rq\mu_(hi,n)}{(1-q)\mu_(li,o)}$
and $b_0=\frac{r(\gamma_{i}+r)}{(1-q)\mu_(li,o)}$ we get $\Omega(i,r)=\frac{\delta_{di}\lambda_i}{a\lambda_{i}+b}$ and\\ $1-\Lambda(i,r)=\frac{b_{1}\lambda_{i}+b_{0}}{a\lambda_{i}+b}.$

We have
\begin{eqnarray*}\Theta(\lambda_{1},...,\lambda_{K})&=&\frac{\sum_{j\in\mathcal{I}}\widetilde{\gamma}_{uj}\Omega(j,r)}{r +\sum_{j\in\mathcal{I}}\widetilde{\gamma}_{uj}(1-\Lambda(j,r))}\\&=&\frac{\frac{\delta_{di}\lambda_i}{a\lambda_{i}+b}+\sum\limits_{\substack{j \neq i}} \widetilde{\gamma}_{uj}\Omega(j,r)}
{\frac{b_{1}\lambda_{i}+b_{0}}{a\lambda_{i}+b}+\sum\limits_{\substack{j \neq i}} \widetilde{\gamma}_{uj}(1-\Lambda(j,r))}\\&=&\frac{[\widetilde{\gamma}_{ui}\delta_{di}+a\sum\limits_{\substack{j \neq i}} \widetilde{\gamma}_{uj}\Omega(j,r)]\lambda_{i}+b\sum\limits_{\substack{j \neq i}} \widetilde{\gamma}_{uj}\Omega(j,r)}
{[b_{1}+ar+a\sum\limits_{\substack{j \neq i}} \widetilde{\gamma}_{uj}(1-\Lambda(j,r))]\lambda_{i}+b_{0}+b(r+\sum\limits_{\substack{j \neq i}} \widetilde{\gamma}_{uj}(1-\Lambda(j,r)))}.
\end{eqnarray*}
The partial derivative of $\Theta(\lambda_{1},...,\lambda_{K})$ in the direction of $\lambda_{i}$ is

$$\frac{\partial \Theta(\lambda_{1},...,\lambda_{K})}{\partial\lambda_{i}}=\frac{\widetilde{\gamma}_{ui}\delta_{di}[b_0+br+b\sum\limits_{\substack{j \neq i}} \widetilde{\gamma}_{uj}(1-\Lambda(j,r))]+(ab_{0}-b_{1}b)\sum\limits_{\substack{j \neq i}} \widetilde{\gamma}_{uj}\Omega(j,r)}
{\left([b_{1}+ar+a\sum\limits_{\substack{j \neq i}} \widetilde{\gamma}_{uj}(1-\Lambda(j,r))]\lambda_{i}+b_{0}+b(r+\sum\limits_{\substack{j \neq i}} \widetilde{\gamma}_{uj}(1-\Lambda(j,r)))\right)^{2}}.$$
Let us show that the quantity $ab_{0}-b_{1}b$ is positive, that is $ab_{0}-b_{1}b> 0.$
Replacing each of $a,b,b_0,b_1$ by its value lead us to

\begin{eqnarray*}ab_{0}-b_{1}b&=&\frac{r(\gamma_i+r)^{2}}
{(1-q)\mu_(li,o)}+\frac{r(\gamma_i+r)(\widetilde{\gamma}_{di}+r)}
{(1-q)^{2}\mu_(li,o)^{2}}q\mu_(hi,n)-\frac{r(\gamma_i+r)(\widetilde{\gamma}_{di}+r)}
{(1-q)^{2}\mu_(li,o)^{2}}q\mu_(hi,n)\\&=&\frac{r(\gamma_i+r)^{2}}
{(1-q)\mu_(li,o)}> 0.
\end{eqnarray*}

 Hence $\frac{\partial \Theta(\lambda_{1},...,\lambda_{K})}{\partial\lambda_{i}}>0.$
The partial derivative $\frac{\partial P_{j}(\lambda_{i})}{\partial\lambda_{i}}=A_{j}\frac{\partial \Theta(\lambda_{1},...,\lambda_{K})}{\partial\lambda_{i}}$, were
$A_{j}=\frac{{\gamma}_{dj}}{(1-q)a_{j}\Gamma(j,r)}-\Lambda(j,r)$ determines the variation of the price. Its sign is entirely determined by the sign of $A_j.$

Now let us investigate to the sign of $A_j.$
We have $A_j>0$ is equivalent to $\frac{{\gamma}_{dj}}{(1-q)\lambda_{j}\mu_(lj,o)}-\Psi(j,r)>0.$ By replacing $\Psi(j,r)$ by its value, a direct calculation gives us $\frac{(\gamma_i+r)(\widetilde{\gamma}_{dj}+r)}
{\lambda_{j}q\mu_(hj,n)}< \gamma_{dj}-\widetilde{\gamma}_{dj}-\frac{(\gamma_j+r)}
{q\mu_(hj,n)}(1-q)\mu_(lj,o).$

$(1).$ If $\gamma_{dj}-\widetilde{\gamma}_{dj}-\frac{(\gamma_j+r)}
{q\mu_(hj,n)}(1-q)\mu_(lj,o)\leq 0,$ then $A_j>0$. Hence the price $P_{j}(\lambda_{i})$ is decreasing.\newline

 Assume that $\gamma_{dj}-\widetilde{\gamma}_{dj}-\frac{(\gamma_j+r)}
{q\mu_(hj,n)}(1-q)\mu_(lj,o)>0$,
We set\newline $\widehat{\lambda}_{j}=\frac{(\gamma_j+r)(\widetilde{\gamma}_{dj}+r)}{(\gamma_{dj}-\widetilde{\gamma}_{dj})q\mu_(hj,n)-(\gamma_j+r)
(1-q)\mu_(lj,o)}$.\newline

$(2).$ For $\lambda_{j}=\widehat{\lambda}_{j}$, the price $P_{j}(\lambda_{i})$ does not depend on the value $\lambda_{i}$, it is a constant function in $\lambda_{i}$.\newline

$(2).$ For $\lambda_{j}> \widehat{\lambda}_{j}$, the price $P_{j}(\lambda_{i})$ is increasing
and for $\lambda_{j}< \widehat{\lambda}_{j}$, the price $P_{j}(\lambda_{i})$ is decreasing. This completes the proof.

\end{proof}

\end{document}